\documentclass[11pt,a4paper]{article}
\usepackage{lmodern} 
\usepackage[T1]{fontenc}

\usepackage{mathtools}
\usepackage{amsmath,amsfonts,amsthm,amssymb,color}
\usepackage[dvipsnames]{xcolor}
\usepackage{fancybox}
\usepackage{hyperref}
\usepackage{algorithm}
\usepackage[labelfont=bf]{caption}
\usepackage[noabbrev,capitalize]{cleveref}
\usepackage[noend]{algpseudocode}
\makeatletter
\def\BState{\State\hskip-\ALG@thistlm}
\makeatother

\newtheorem{theorem}{Theorem}[section]
\newtheorem{corollary}[theorem]{Corollary}
\newtheorem{lemma}[theorem]{Lemma}
\newtheorem{observation}[theorem]{Observation}
\newtheorem{proposition}[theorem]{Proposition}
\newtheorem{claim}[theorem]{Claim}
\newtheorem{fact}[theorem]{Fact}
\newtheorem{assumption}[theorem]{Assumption}

\theoremstyle{definition}
\newtheorem{definition}[theorem]{Definition}

\def\eps{\varepsilon}
\def\Real{\mathbb{R}}
\def\poly{\operatorname{poly}}

\def\Bar{\overline}
\def\Hat{\widehat}
\DeclareMathOperator{\Diag}{\mathbf{Diag}}

\DeclareMathOperator{\dom}{dom}
\DeclareMathOperator{\Cong}{cong}

\DeclareMathOperator{\str}{str}
\DeclareMathOperator{\grad}{\nabla}

\newcommand\ftil{\tilde{f}}

\newcommand\Otil{\widetilde{O}}

\def\bb{\mathbf{b}}
\def\bc{\mathbf{c}}
\def\bd{\mathbf{d}}

\def\bf{\mathbf{f}}

\def\bv{\mathbf{v}}
\def\bs{\mathbf{s}}
\def\bt{\mathbf{t}}

\def\bu{\mathbf{u}}
\def\bx{\mathbf{x}}
\def\by{\mathbf{y}}
\def\bz{\mathbf{z}}

\def\bA{\mathbf{A}}
\def\bB{\mathbf{B}}
\def\bC{\mathbf{C}}

\def\bL{\mathbf{L}}

\def\bR{\mathbf{R}}

\def\bX{\mathbf{X}}

\def\CC{\mathcal{C}}
\def\DD{\mathcal{D}}
\def\EE{\mathcal{E}}
\def\GG{\mathcal{G}}
\def\HH{\mathcal{H}}
\def\MM{\mathcal{M}}
\def\PP{\mathcal{P}}
\def\RR{\mathcal{R}}
\def\OO{\mathcal{O}}
\def\WW{\mathcal{W}}

\newcommand{\norm}[1]{\left\lVert#1\right\rVert}
\newcommand{\Abs}[1]{\left\vert#1\right\vert}
\newcommand{\T}[1]{#1 ^\top}
\newcommand{\pinv}[1]{#1 ^\dagger}
\newcommand{\mb}[1]{\mathbf{#1}}

\newcommand{\vol}{{\hbox{\textnormal{vol}}}}
\newcommand{\cost}{{\hbox{\textnormal{cost}}}}

\newcommand{\argmin}{\mathop{\mathrm{argmin}}}
\newcommand\rea{\mathbb R}
\newcommand{\vone}{\boldsymbol{\mathbf{1}}}

\begin{document}

\title{$\ell_2$-norm Flow Diffusion in Near-Linear Time}
\author{
		Li Chen\\
		Georgia Tech\\
		\href{mailto:lichen@gatech.edu}{lichen@gatech.edu}
		\thanks{Research supported by IDEaS-TRIAD Research Scholarship}
        \and 
        Richard Peng\\
		Georgia Tech\\
		\href{mailto:rpeng@cc.gatech.edu}{rpeng@cc.gatech.edu}
        \and
        Di Wang\\
        Google Research\\
        \href{mailto:wadi@google.com}{wadi@google.com}
}

\maketitle

\begin{abstract}
Diffusion is a fundamental graph procedure and has been a basic building block in a wide range of theoretical and empirical applications such as graph partitioning and semi-supervised learning on graphs. In this paper, we study computationally efficient diffusion primitives beyond random walk.

We design an $\widetilde{O}(m)$-time randomized algorithm for the $\ell_2$-norm flow diffusion problem, a recently proposed diffusion model based on network flow with demonstrated graph clustering related applications both in theory and in practice. Examples include finding locally-biased low conductance cuts. Using a known connection between the optimal dual solution of the flow diffusion problem and the local cut structure, our algorithm gives an alternative approach for finding such cuts in nearly linear time.

From a technical point of view, our algorithm contributes a novel way of dealing with inequality constraints in graph optimization problems. It adapts the high-level algorithmic framework of nearly linear time Laplacian system solvers, but requires several new tools: vertex elimination under constraints, a new family of graph ultra-sparsifiers, and accelerated proximal gradient methods with inexact proximal mapping computation.

\end{abstract}
\newpage

\tableofcontents
\newpage

\section{Introduction}


Graphs are among the most prevalent data representations when it comes to modeling real-world networks or capturing relationships among entities. Thus, efficiently computing on graphs is a central topic in combinatorial optimization with wide applications in operations research, data mining, and network science etc. In this paper we look at {\em diffusion}, which informally is the generic process of spreading mass among vertices by sending mass across edges, and it can be effective at revealing interesting structures of the graph or learning latent representations of vertices. Diffusion has been a basic building block in the study of various fundamental algorithmic problem on graphs, most notably clustering~\cite{SpielmanT13,ACL06, chung2007-pagerank-heat, AP09,AL08_SODA, OZ14, OrecchiaSV12,SaranurakW18}, and it has also been successful at various graph learning tasks such as node classification and embedding~\cite{BC18, IbrahimG19, KWG19, LG20, FWY20, WFHM2017,HHSLB20, BKP20}. Consequently it is of great interests both theoretically and empirically to further the development of simple diffusion primitives that can be efficiently computed over massive graphs.

There has been a long line of research on diffusion, and historically the most popular ones are predominantly spectral methods based on the dynamics of random walk over graphs~\cite{LS90,pagerank,ACL06,SpielmanT13,chung2007-pagerank-heat,AP09,ALM13}. Although random walk based diffusion methods have seen wide applications due to their conceptual simplicity and practicality, they can spread mass too aggressively, and it is known in theory and in practice that spectral methods can have unsatisfactory performances when structural heterogeneities exist, and thus are not robust on real-world graphs constructed from very noisy data~\cite{guatterymiller98,Luxburg07,LLDM09_communities_IM,Jeub15}. More recently there has been progress on diffusion methods based on the dynamics of network flow, which have demonstrated improved theoretical guarantees over spectral methods on graph partitioning~\cite{OZ14,HRW17,SaranurakW18}, as well as better empirical results on learning tasks in practice~\cite{WFHM2017,FWY20}. On the computational side, although some diffusion methods are strongly local, i.e. the running time depends on the size of the output cluster instead of the size of the whole graph, in general the time to compute these diffusions is either not nearly linear in the (local or global) graph size or depends polynomially on quantities such as the conductance of the graph, and thus can be slow when one wants to explore large region of a graph or the conductance is small. 

Motivated by these observations, we study computationally efficient diffusion primitives beyond random walk. In particular, we look at the $\ell_p$-norm flow diffusion model proposed in~\cite{FWY20}, which offers the similar conceptual simplicity of random walk based diffusion models, and the authors also demonstrated in their paper that the $\ell_p$-norm flow diffusion can give theoretical guarantees for graph clustering and have better empirical performance than spectral diffusion at practical learning tasks such as community detection. In this work, we tackle the computational problem, and show the $\ell_p$-norm flow diffusion problem can be computed to high accuracy in nearly linear time for $p=2$. Before presenting our main result formally in~\cref{sec:IntroResult}, we first discuss the more generic setup for flow diffusion problems briefly for the background.

\subsection{Flow Diffusion}
\label{sec:FlowDiffusion}
The flow diffusion comes in a succinct generic optimization formulation below
\begin{equation}
    \label{eq:general_diffusion}
    \min_{\bf \in \Real^E} \cost(\bf)
    \text{ s.t. } \T{\bB}\bf \leq \bd,
\end{equation} 
where $\bB$ is the edge incidence matrix of a graph, $\bd \in \Real^V$ is the vertex capacity vector with $\bd^{\top} \vone \geq 0$, and $\cost(\cdot)$ is a cost function on flows. The {\em $\ell_p$-norm flow diffusion} takes $\cost(\bf)=\norm{\bf}_p$ as the cost function for some $p\in(0,\infty)$.

To see the intuitive combinatorial interpretation behind flow diffusion, we go back to the fundamental problem of network flow, which can be written in a similar generic form as follows.
\begin{equation*}
    \min_{\bf \in \Real^E} \cost(\bf)
    \text{ s.t. } \T{\bB}\bf = \bb,
\end{equation*}
with $\bb \in \rea^{V}$ (with $\bb^{\top} \vone = 0$) being the demand vector. The network flow problem asks for the most cost efficient way of routing a given demand $\bb$, whereas in flow diffusion, the solution can route any demand $\bb=\T{\bB}\bf$ as long as $\bb\leq \bd$. It is straightforward to see that flow diffusion generalizes network flow, since routing a specific demand $\bb$ can be enforced by using $\bd=\bb$ in~\eqref{eq:general_diffusion}. 

For a concrete application, in~\cite{FWY20} the authors looked at $\ell_p$-norm flow diffusion in the context of finding locally-biased graph clustering.
One can consider the vector $\bd$ as $\bt - \bs$ for $\bs,\bt \in \rea^{V}_{\geq 0}$ and $\bs^{\top} \vone\leq \bt^{\top} \vone$, so the constraint in~\eqref{eq:general_diffusion} becomes $\T{\bB}\bf +\bs\leq \bt$.
One can view $\bs(v)$ as the source mass placed on node $v$, where $\bs$ typically has a small support, e.g. on a small set of given seed nodes, with each of such node having a relatively large amount of mass, or informally "high density".
$\bt(v)$ is the sink capacity of node $v$, where $\bt$ has larger support but smaller infinity norm (i.e. "lower density") comparing to $\bs$, e.g. $\bt(v)=\deg(v)$ $\forall v$.
The flow diffusion aims to "spread out" the source mass from the high density state $\bs$ to any lower density state allowed by $\bt$ using the most cost-efficient network flow. One distinction of~\eqref{eq:general_diffusion} is its compact optimization formulation which enables direct analysis of its optimal solution (independent of any algorithm), whereas almost all other diffusion methods are defined with the step-by-step rules of how to spread mass, and thus the analysis must crucially rely on the specific dynamics of the underlying procedures. In particular, the optimal solution to the $\ell_p$-norm flow diffusion problem (or more accurately its dual to be introduced shortly) can be used to find local low conductance cuts with provable approximation guarantees (parameterized by $p$). For more details, we refer readers to~\cite{FWY20}.

We note that the generic idea of using sink capacities with larger sum than the total source mass to allow flexibility in demand long predates the formulation of~\eqref{eq:general_diffusion}, and has been exploited specifically in local clustering (e.g.~\cite{OZ14}).
Indeed, when $\cost(\bf)=\norm{\bf}_{\infty}$ the diffusion problem can be reduced to standard s-t max-flow by adding a super-sink $t$ and an edge from each node to $t$ with edge capacity equal to the node's sink capacity.
However, the reduction to network flow doesn't hold in general, and flow diffusion beyond the special case of $\norm{\bf}_{\infty}$ has been relatively under-explored. 

Moreover, diffusion models based on random walk (e.g.~\cite{pagerank,chung2007-pagerank-heat}) are typically linear operators, while flow diffusion naturally introduces non-linearity. Note this is the case even in $\ell_2$-norm flow diffusion whose counterpart network flow problem (i.e. electrical flow) is a linear operator. Introducing non-linearity in the diffusion process has proved to be useful at improving the classification accuracy empirically~\cite{BC18,ID19,FWY20}, and this further motivates our study of simple and practical diffusion algorithms beyond this special case.

\subsection{Our Results}
\label{sec:IntroResult}
We focus on the case with $\cost(f)$ being the electrical energy of the flow $\bf$ in~\eqref{eq:general_diffusion}. 
Given an undirected graph $G=(V, E)$ with conductance vector $\bc\in\rea^{E}_{\geq 0}$ (i.e., edge $e$ has resistance $1/c(e)$), vector $\bd \in \rea^{V}$ with $\T{1}\bd \ge 0$, and denote the Laplacian Matrix of $G$ as $\bL(G)=\T{\bB}\bC\bB$ where $\bB\in\rea^{E\times V}$ is the edge incidence matrix and $\bC$ is the diagonal matrix of $\bc$, the $\ell_2$-norm flow diffusion and its dual are as follows
\begin{equation}
\label{L2_primal}
    \min_{\bf \in \Real^E} \sum_{e \in E} \frac{f_e^2}{2c(e)}
    \qquad \text{ s.t. } \T{\bB}\bf \le \bd,
\end{equation}
\begin{equation}
\label{L2_dual}
    \min_{\bx \in \Real^n} \frac{1}{2} \T{\bx}\bL(G)\bx + \T{\bd}\bx 
    \qquad \text{ s.t. } \bx \ge 0.
\end{equation}
Similar to the canonical relationship between the flow and vertex potential in electrical flow, the optimal dual $\bx^*$ in our case also directly induce the optimal primal solution to the flow diffusion problem via $\bf^* = -\bC^{-1}\bB\bx^*$. We design the first nearly linear time algorithm to solve the dual problem~\eqref{L2_dual} (and thus the primal) to high accuracy. More formally,
\begin{theorem}
\label{thm:L2Main}
Given an instance of~\eqref{L2_dual} on graph $G$ of $n$ vertices and $m$ edges, error parameter $\varepsilon > 0$, with high probability Algorithm~\ref{algo:Main} computes $\bx$ whose objective is within $(1+\varepsilon)$ factor of the optimum of~\eqref{L2_dual} in running time $O(m\log^8 n \log (1 / \varepsilon))$.
\end{theorem}

An immediate application of our result is to find locally-biased low conductance cuts, where the (edge) conductance $\phi(S)$ of a cut $S\subsetneq V$ is defined as 
\[
\phi(S) = \frac{\sum_{e\in E(S,\bar{S})} c(e)}{\vol(S)}
\]
where the volume $\vol(S)$ is the sum of (weighted) degrees of nodes in $S$.
This comes as a corollary of the result (Theorem $2$ in~\cite{FWY20}) connecting the optimal solution to the dual problem with low conductance cuts around a given set of seed nodes. Intuitively, if we put flow mass on the seed nodes and ask for the most energy efficient way to spread the mass to a larger region, it will be relatively easy to spread within a well-connected local neighborhood, while taking high energy to push mass out of a low conductance bottleneck.
Thus, the dual solution will naturally reveal the local cut structure, and algorithmically this can be done easily by a sweep cut on the dual solution.
\begin{corollary}
\label{coro:clusterMain}
Given a set of seed nodes $S$ in a graph $G$ of $n$ nodes and $m$ edges, suppose there exists a cluster $C$ such that
\begin{enumerate}
    \item $\vol(S\cap C)\geq \alpha\vol(C)$ for some $\alpha\in (0,1]$,
    \item $\vol(S\cap C)\geq \beta\vol(S)$ for some $\beta\in (0,1]$.
\end{enumerate}
Then there is an algorithm that with high probability finds a cut $\tilde{C}$ in  $O(m\log^8 n \log (1 / \varepsilon))$ time such that $\phi(\tilde{C})\leq O\left(\frac{\sqrt{\phi(C)}}{\alpha\beta}\right)$
\end{corollary}
The conductance guarantee follows directly from Theorem $2$ in~\cite{FWY20}, which proves more generally that for $\ell_p$-norm flow diffusion with $p\in(1,\infty)$, the sweep cut on the optimal dual solution finds a cut $\tilde{C}$ such that $\phi(\tilde{C})\leq O\left(\frac{\phi(C)^{(p-1)/p}}{\alpha\beta}\right)$. 

For comparison~\cite{FWY20} gives a strongly local method with worst-case running time guarantee of $O(\vol(C)^3D^2\log(1/\varepsilon))$ in the $p=2$ case where $D$ is the maximum degree of explored nodes, which can be as slow as $O(m^3n^2\log(1/\varepsilon))$ when $\vol(C)=\Theta(m)$ and $D=\Theta(n)$.

\subsection{Related Works}
\label{sec:RelatedWorks}
Diffusion is one of the most basic graph exploration techniques and has connections to various graph problems both in theory and in practice. We discuss some of these topics in this section.

\textbf{Graph Partitioning.}
Graph partitioning is a fairly broad area, and in our discussion we restrict our attention to the more specific problem of finding low conductance cuts in graphs, which along with its many variants are NP-Hard and of central object of study for approximation algorithms. Results along this direction can come in variants such as expander decomposition, where a graph is partitioned into clusters with large internal conductance (i.e. expanders) and small number of inter-cluster edges, or balanced separator, e.g. finding low conductance cut whose smaller side is at least a constant fraction of the total graph size. Here, the conductance of a cluster $C \subseteq V$ is defined as $\phi(C) \coloneqq |E(C, \Bar{C})| / \min\{\vol(C), \vol(\Bar{C})\}$. Flow and random walk methods have long been used to find good graph decompositions. For example, inspired by the classical result of Lovasz and Simonovitz~\cite{LS90}, Spielman and Teng~\cite{SpielmanT13} demonstrated that a constant fraction of vertices in a low conductance cut are good random walk starting points to find a good cut using diffusion. Their result has been improved by subsequent work~\cite{ACL06,AP09,OrecchiaV11,GharanT12,OrecchiaSV12} to achieve better approximation ratio and running time. These methods, similar to spectral partitioning methods in general~\cite{AlonM85,SINCLAIR198993}, suffer an intrinsic quadratic approximation error, informally known as the Cheeger barrier due to Cheeger's Inequality~\cite{Cheeger69}. That is, if there exists a cut of conductance $\phi^*$, these methods guarantee to find a cut with conductance no larger than $O(\sqrt{\phi^*})$. The running time is typically nearly linear in the size of the graph and with polynomial dependence in $\phi$.~\cite{OrecchiaSV12} achieves nearly linear time algorithm independent of $\phi$ via more advanced techniques including solving a series of Laplacian systems and \emph{Semidefinite programming} (SDP) rounding. 
Beyond spectral methods,~\cite{LeightonRao} uses LP relaxation to give $o(\log n)$ approximation on sparsest cut, and~\cite{ARV} achieves the best approximation ratio of $O(\sqrt{\log n})$ using SDP relaxation, whose running time has been subsequently optimized in~\cite{Sherman09,AroraHK10,AroraK07}. Despite having good theoretical guarantees, these algorithms are not easy to implement or even capture in a principled heuristic. The methods of~\cite{KRV,OSVV} combined the techniques of random walk and flow algorithms to give a framework of \emph{cut-matching game}, which give poly-logarithmic approximation ratio using a poly-logarithmic number of approximate max flow computations. The framework has inspired other graph partitioning algorithms~\cite{Sherman09,SaranurakW18,ChuzhoyGLNPS19:arxiv}.

\textbf{Local Clustering.}
A closely related problem (and sometimes a subroutine) of (global) graph partitioning is the \emph{local clustering} problem, which aims to find low conductance cluster local to a given subset of vertices $A$. The problem is first proposed and studied in \cite{SpielmanT13}, and is applied broadly in the design of fast (global) graph algorithms and also semi-supervised learning on graphs. 
Some of the results are weakly local in the sense that only the result cluster is biased toward the given seed set while the algorithm's running time still depends on the size of the entire graph, e.g.~\cite{MOV12_JMLR}. On the other hand, strongly local algorithms only touches the local region around the seed set, and run in time proportional to the size of the local cluster. Strongly local methods are predominantly based on the framework of exploring the graph starting from the seed nodes using various diffusion methods. Similar to the global partitioning case, random walk based diffusion methods also suffer the quadratic loss in the local clustering setting.
 On the other hand, network flow algorithms have also been adapted as strongly local diffusion methods to find good local cluster or to locally improve the cut quality (i.e. conductance) of given cut~\cite{OZ14, WFHM2017}.
They achieve better approximation ratio of $O(\phi^* polylog(n))$ by exploiting the canonical flow-cut duality.
However, diffusion methods based on combinatorial flow algorithms (e.g. Dinic's algorithm~\cite{Dinic70} in~\cite{OZ14} and push-relabel~\cite{GoldbergT88} in~\cite{WFHM2017}) are generally considered more difficult to understand and implement in practice due to the more complicated underlying dynamics.

We note that graph partitioning and local clustering algorithms are widely used in the design of efficient graph algorithms, for example graph sparsification~\cite{SpielmanT11:journal}, dynamic data structures~\cite{NanongkaiSW17}, minimum cut~\cite{KT15,HRW17,Li20Mincut}, etc.

\textbf{Electrical network flow and Laplacian Solver.}
From a technical point of view, our $\ell_2$-norm flow diffusion algorithm is inspired by the seminal result of Spielman and Teng~\cite{SpielmanTengSolver:journal} for solving Laplacian linear systems in nearly linear time and the subsequent progress (e.g.~\cite{KMP,KoutisMP10,KyngS16,JS20}) along this line of research (a.k.a the Laplacian paradigm~\cite{Teng10:survey}), which have served as a building block for the breakthroughs on many fundamental graph problems. Our algorithm at a high level uses a similar approach, yet the lack of given demand in flow diffusion poses unique challenges and requires novel techniques. 

There has been some recent progress on high accuracy algorithms to solve $\ell_p$-norm network flow problem. Both our $\ell_2$-norm diffusion result and $\ell_p$-norm network flow exploit ideas developed in the Lapalcian framework, and are natural non-linear variants of the electrical flow problem. Kyng et al.~\cite{KPSW19} obtained an almost-linear time high-accuracy algorithm for the uniform-weighted $\ell_p$-norm flow problem for $p\in[2,\infty)$, which has thus been improved and also generalized to weighted graphs~\cite{AdilS20,ABKS21}.
Using these results as a black box, Axiotis et al.~\cite{AxiotisMV20}, and Kathuria et al.~\cite{KLS20} obtained the first almost-$O(m^{4/3})$-time algorithm for the unit-capacitated max flow problem (i.e. $p=\infty$).



\textbf{Learning on Graphs}
There has been a surge of interest in performing machine learning tasks on graphs in recent years. The field of \emph{learning on graphs} has produced a fruitful of empirical results and practical applications. We refer readers to \cite{HamiltonGraphLearningBook} for a comprehensive survey.

Among the practical tasks, the node classification problem is of central importance.
The problem asks to classify vertices in a graph where labels are available for a subset of vertices.
Graph neural network (GNN) based approaches~\cite{KW16GCN,B18GNN,Wu20GNN} have been the state of the art in terms of accuracy, but the lack of efficient training algorithms and huge model complexity yield challenges when applying them on large graphs. Some recent results turn to diffusion based methods due to its simplicity and scalability.
\cite{BC18} showed advantages of diffusion over node embedding methods and graph convolutional networks~\cite{KW16GCN}.
\cite{KWG19} and \cite{BKP20} incorporated diffusion with GNN and get faster training time than other GNN-based methods.
\cite{IbrahimG19} showed the advantage of nonlinear diffusion process defined by $p$-Laplacian. When node features are presented, \cite{HHSLB20} achieved the state of the art result by using diffusion as a post-process routine independent from the learning procedure.





\subsection{Future Directions}
\label{sec:Future}
There are several natural next steps to further the study of flow diffusion. 

\textbf{Efficient Algorithms for the general $\ell_p$-norm flow diffusion.}
\cite{FWY20} has shown empirical evidence that using larger $p$ (e.g. $4$ or $8$) leads to improved clustering results over $p=2$. Moreover as mentioned after Corollary~\ref{coro:clusterMain}, $\ell_p$-norm flow diffusion gives local cut with conductance approximation guarantee of $O(\phi^{(p-1) / p})$, which is similar to the quadratic approximation of spectral diffusion starting at $p=2$, and approaches optimality as $p$ gets large. Thus, design efficient algorithm for $\ell_p$-norm flow diffusion beyond the case of $p=2$ has immediate impact both in theory and in practice. 
The direction is promising given the analogous development of fast algorithms for the $\ell_p$-norm flow problem~\cite{KPSW19, AdilS20, ABKS21}, and it is plausible that some techniques such as adaptive preconditioning~\cite{KPSW19}, sparsification~\cite{ABKS21}, and iterative refinement~\cite{AKPS19, AdilS20} developed in the context of $\ell_p$-norm network flow can be extended to the diffusion case.

\textbf{Finding balanced cut or global graph decomposition.}
Beyond local clustering guarantees, it is interesting to build upon the $\ell_p$-norm flow diffusion primitive to give (global) graph partitioning results such as balanced separator. This is analogous to how random walk based diffusion can serve as the building block for spectral partitioning, e.g.~\cite{SpielmanT13}. 

A graph partitioning result analogous to the smooth $O(\phi^{(p-1)/p})$ approximation guarantee in the local clustering case~\cite{FWY20} would be exciting as it bridges the guarantees given by the different approaches of spectral and flow in a unified framework.

\textbf{Improved Algorithm for $\ell_2$-norm flow diffusion.}
Similar to how the initial nearly linear time algorithm of Spielman and Teng for electrical flow is made simpler, faster and more practical, we expect similar progress can be achieved in the flow diffusion case. It is also interesting to design strongly local algorithms for flow diffusion with nearly linear dependence on the output cluster size, which can have practical impact when the size of the graph is massive comparing to the local cluster, and can also be useful for the design of nearly linear graph decomposition algorithm.



\section{Technical Overview}
\label{sec:TechOverview}
At a high level, we use the approach of a (preconditioned) numerical method, which is similar to many other almost-linear time algorithms for graph optimization problems~\cite{SpielmanTengSolver:journal, Sherman13, KelnerLOS14, Peng16, Sherman17a, Li20, ASZ20, KPSW19}.
Among these, the problem closest to ours is the electrical flow problem which is equivalent to solving graph-structured linear systems~\cite{SpielmanTengSolver:journal, KoutisMP10, KMP, JS20}.
It has similar form as Problem~\eqref{L2_primal} by replacing the inequality constraint with equality, and its dual is just Problem~\eqref{L2_dual} without the non-negativity constraint.

The general idea of the numerical approach is to iteratively update a current solution so that the gap to optimum in objective value is reduced geometrically (a.k.a iterative refinement). In each iteration, we look at a residual problem which captures the impact of an update (to the current solution) on the original objective, and solving the residual problem approximately would give us an update which improves the objective value substantially. In particular, we solve a slightly generalized form of Problem~\eqref{L2_dual} (see~\ref{sec:generalized_diffusion}) so that the residual problem is also in the same form, and we can solve the residual problem recursively.

Similar to the Spielman-Teng Laplacian solver, we recursively solve a residual problem of size $m$ by solving a series of smaller sized problems. In particular, we first precondition the original instance to get instances that well approximate the original instance (see Section~\ref{sec:generalized_diffusion} for formal definitions) and also have special graph structures so that we can then substantially reduce the size of the graph via vertex elimination. A critical component of this scheme is to map solutions back and forth between the smaller instance and the original one so a good solution on the smaller instance can be transformed to a good solution of the original instance. Informally, our preconditioner incurs a total (multiplicative) factor $k$ loss between the original instance and the preconditioned instance, and the vertex elimination reduces the instance size to $\Otil(m/k)$ without any further loss in solution quality. Together with a fast numerical routine based on proximal accelerated gradient descent (AGD), we can reduce the original problem of size $m$ to solving $\Otil(\sqrt{k})$ ones of size $\Otil(m/k)$ for some $k$. Our preconditioner construction and vertex elimination steps take $\Otil(m)$, so we can solve an instance of size $m$ in time $T(m)$ which satisfies the following recurrence:
\begin{align*}
    T(m) = \Otil\left(\sqrt{k} \cdot m + \sqrt{k}\cdot  T\left(\Otil\left(\frac{m}{k}\right)\right)\right).
\end{align*}
Taking $k = \poly\log n$, we can bound $T(m)$ by $\Otil(m)$.


We give a more detailed discussion of the aforementioned technical components in the rest of the section. 
\subsection{Comparison with Laplacian System Solvers}
For readers familiar with the classical result of Spielman and Teng~\cite{SpielmanTengSolver:journal}, the Laplacian system solver has the similar components of preconditioner, vertex elimination and a numerical method on top of these. We start with a high level comparison to point out some of the unique challenges in our problem. 

\begin{table}
\begin{center}
\begin{tabular}{|c|c|c|}
\hline
 & Spielman and Teng~\cite{SpielmanTengSolver:journal} & Our Algorithm \\ 
\hline
Vertex Elimination & Gaussian Elimination & VWF Maintenance~[\cref{lemma:vtxElimination}] \\
\hline
Preconditioner & Spectral Ultrasparsifier & $j$-Tree Sparsifier [\cref{lemma:JTreeSparsify}] \\  
\hline
Numerical Routine & Preconditioned Chebyshev Method & Proximal Gradient Method~[\cref{lemma:ProxAGD}] \\  
\hline
\end{tabular}
\end{center}
\caption{\label{tab:STframework}Comparison between the Laplacian solver and our diffusion algorithm}
\end{table}

Several technical obstacles arise when trying to adapt the approach of nearly linear time Laplacian system solvers to the flow diffusion case.
The first one is the lack of the vertex elimination scheme.
It is a critical component in Laplacian solver and many of the similar recursive preconditioning-based algorithms.
Commonly, Laplacian solvers eliminates vertices of degree $1$ or $2$ via Gaussian elimination. When mapping between solutions of the reduced instance and the original instance, Values of the eliminated vertices can be recovered easily via a linear combination of its neighboring vertices. Gaussian elimination is not applicable in our setting due to the existence of the non-negative constraint $\bx \ge \mb{0}$.
Moreover, when eliminating some degree $1$ vertex, the resulting problem is not in the form of Problem~\eqref{L2_dual}.
We resolve the issue by extending the problem space and designing elimination routine to only remove vertices of degree $1$.

As a consequence, the second obstacle is that existing preconditioners don't have the graph structures that allow us to reduce the graph size substantially just from eliminating degree $1$ vertices. Laplacian system solvers such as \cite{SpielmanTengSolver:journal} uses \emph{spectral ultra-sparsifier} as preconditioner.
A spectral ultra-sparsifier is a subgraph of the original graph $G$ with few edges, i.e. $n-1 + o(m)$.
Eliminating degree $2$ vertices is crucial in reducing the problem size on a spectral ultra-sparsifier.
Consider the example that the spectral ultra-sparsifier is a spanning cycle.
Every vertex has degree $2$ and cannot be eliminated in our case.
We resolve the issue by introducing a new family of preconditioners called \emph{$j$-tree sparsifier}.
A $j$-tree is a union of $j$ rooted trees and a graph on the $j$ roots.
After iteratively eliminating degree $1$ vertices, the resulting graph has at most $j$ vertices.
We show how to construct a $\Otil(m/k)$-tree that is also a preconditioner of quality $k$ for any $k > 0$.

The third one is about the numerical primitives.
Laplacian system solvers use schemas such as preconditioned conjugate gradient descent or preconditioned Chebyshev method. Less is known on generalizing them to constrained settings such as Problem~\eqref{L2_dual}.
Proximal gradient methods~\cite{Nesterov13, BeckT09, Tseng08} is a family of optimization methods solving problems composed of a smooth convex term and a non-smooth term.
Problem~\eqref{L2_dual} falls under such category since it consists of a smooth term $\bx^\top \bL(G) \bx$ and a non-smooth term $\bd^\top \bx$ and constrain $\bx \ge \mb{0}$.
We adapt these methods and analyze the convergence rate in our setting.

The difference between our algorithm and the standard Laplacian solver~\cite{SpielmanTengSolver:journal} is summarized as \cref{tab:STframework}.
Next, we present our algorithm in more detail.






\subsection{Vertex Elimination via VWF maintenance}

The goal of vertex elimination is to efficiently eliminate vertices of degree $1$ one by one.
In addition, one can efficiently recover the solution given any solution to the reduced problem instance.
The whole process of elimination and recovery should take $\Otil(m)$-time where $m$ is the problem size.

By eliminating a degree $1$ vertex $u$ adjacent only to $v$, we reduce solving Problem~\eqref{L2_dual} to solving the following:
\begin{align*}
    \min_{\bz \in \Real^{V\setminus \{u\}}_{\ge 0}} \frac{1}{2}\bz^\top \bL(G') \bz + \sum_{s \in V \setminus \{u\}} d_s z_s + \frac{1}{2}c(uv)(z_v - M(z_v))^2 + d_u M(z_v),
\end{align*}
where $G'$ is obtained by removing $u$ from $G$, and $M$ is some mapping that recovers the value of $u$ given the value of its only neighbor $v$.
In the unconstrained case, $M$ is an affine function and the reduced problem has similar form of $\min_{\bz} (1/2) \bz^\top \bL(G') \bz + \bd'^\top \bz$.

The main difficulty in eliminating degree $1$ vertex in Problem~\eqref{L2_dual} is that such mapping $M$ is not affine.
The reduced problem is no longer a $\ell_2$-norm flow diffusion problem.
For example, we decrease the value of $v$ gradually.
The optimal value of $u$ decreases as well since the Laplacian term in Problem~\eqref{L2_dual} encourages similar values between neighbors.
However, the optimal value at $u$ stops at 0 when the value of $v$ is smaller than some threshold.
We resolve the issue by extending the problem space.

A function $f:\Real \to \Real$ is a \emph{Vertex Weighting Function} (VWF) if $f$ is convex, piece-wise quadratic, and has concave continuous derivative.
An example of VWF is the family of linear functions $f(x) = ax$ for any $a \in \Real$.
For simplicity of the presentation, we use $|f|$ or the \emph{size of $f$} to denote the number of pieces of a VWF $f$.
Detailed definition and basic properties of VWFs are presented in \cref{sec:prelim}.

Instead of $\bd^\top \bx$, our algorithm solves a generalized version of Problem~\eqref{L2_dual} as follows:
\begin{align}
\label{gen_L2_dual}
    \min_{\bx \ge \bb} \frac{1}{2} \sum_{e=uv \in E} c(e)(x_u - x_v)^2 + \sum_{u \in V} f_u(x_u),
\end{align}
where $\bb \le \mb{0}$ indicates the lower-bound on each vertex, and each $f_u$ is a VWF .
Clearly, Problem~\eqref{L2_dual} is a special case of Problem~\eqref{gen_L2_dual} by defining $\bb = \mb{0}$ and $f_u(x) = d_u x$ for every vertex $u$.

When eliminating a degree 1 vertex $u$ adjacent only to $v$, Problem~\eqref{gen_L2_dual} is reduced by removing variable $x_u$ and replacing $f_v$ with
\begin{align*}
    f_v^{\mathrm{new}}(x) = f_v(x) + \min_{y \ge b_u} \frac{1}{2} c(uv) (x - y)^2 + f_u(y).
\end{align*}
We use $f_v^{\mathrm{new}}(x)$ to encode the optimal contribution of vertex $u$ to the objective value of Problem~\eqref{gen_L2_dual} before the elimination.
In \cref{sec:VWFDS}, we can show that $f_v^{\mathrm{new}}(x)$ is another VWF whose representation can be obtained efficiently from the one of $f_v$.

Moreover, $f_v^{\mathrm{new}}(x)$ has at most $|f_v| + |f_u| + 1$ pieces.
The total number of pieces can increase by $1$ after eliminating $1$ vertex.
Thus, after eliminating $k$ vertices, the total size of VWFs increases by $k$ despite the number of VWF is decreased by $k$.
This behavior is problematic in reducing the problem size effectively.
We resolve the issue by compressing VWFs to size of $O(\log n)$.

Specifically, we show that, given any instance of Problem~\eqref{gen_L2_dual}, it can be solved in high-accuracy by solving $O(\log n)$ instances whose VWF are all of size $O(\log n)$.
The key idea is to approximate the VWF $f$ by one of size $O(\log n)$, say $g$.
This is done by rounding its split points to the nearest integer power of $1.1$.
In \cref{sec:CompressVWF}, we can show that
\begin{align*}
    2 f\left(\frac{x}{2}\right) \le g(x) \le f(x), \forall x \in \dom(f).
\end{align*}
This criteria of approximation is motivated by our use of \emph{iterative refinement}~\cite{AKPS19, KPSW19} to compute highly-accurate solution.

Given an instance of Problem~\eqref{gen_L2_dual} with $m$ edges, and total VWF size $S$, the run-time of our algorithm is denoted as $T(m, S)$.
Our vertex elimination routine yields the following:
\begin{align*}
    T(m, S) = \Otil\left(O(\log^2 n)(m + S + T(k, O(k \log n)))\right),
\end{align*}
if the remaining graph after elimination has at most $k$ vertices and edges.

\subsection{Preconditioner Construction: J-Trees}
Given a graph $G$, a preconditioner of {\em quality $k$} is another graph $H$ whose Laplacian matrix is similar to $G$'s, i.e. $\bL(G) \preceq \bL(H) \preceq k \bL(G)$ holds where $\bA \preceq \bB$ indicates that $\bB - \bA$ is positive semidefinite.
In our definition, $H$ does not have to be a subgraph of $G$.

\emph{Spectral Ultra-Sparsifiers} are used in \cite{SpielmanTengSolver:journal} to precondition the Laplacian system.
Besides being a preconditioner of quality $k$, a \emph{Spectral Ultra-Sparsifiers} is a subgraph of $G$ with $n-1+\Otil(m/k)$ edges.
One can view a spectral ultra-sparsifier as a spanning tree of $G$ plus $\Otil(m/k)$ off-tree edges.
Solving Laplacian system on it can be reduced to one of size $\Otil(m/k)$.
The reduction is done via repeatedly eliminating vertices of degree 1 or 2.

In our setting, we are restricted to eliminating vertices of degree 1.
It prevents us from using spectral ultra-sparsifiers and its later improvements~\cite{SpielmanTengSolver:journal, KoutisMP10, KollaMST10}.
For example, consider the spectral ultra-sparsifer being a spanning cycle, which has $n-1 + 1$ edges.
One cannot eliminate any vertex from it since every vertex has degree 2.

We resolve the problem by introducing a new family of preconditioners called \emph{$j$-tree Sparsifiers}.
A \emph{$j$-tree} is an union of $j$ rooted trees and a graph supported on their roots.
This concept is first introduced in \cite{Madry10} for approximating cuts in graphs.
The advantage of using $j$-tree is that the remaining graph is induced by the $j$ roots after eliminating degree 1 vertices one at a time.
In \cref{sec:JTreeSparsify}, we present efficient construction of $\Otil(m/k)$-tree sparsifiers being a preconditioner of quality $k$ given any $k > 0$.

Combining with the vertex elimination routine and the VWF compression scheme, solving Problem~\eqref{gen_L2_dual} on a $\Otil(m/k)$-tree by our algorithm is done in
\begin{align*}
    T\left(n-1+\Otil\left(\frac{m}{k}\right), S\right) = \Otil\left(m + S + T\left(\Otil\left(\frac{m}{k}\right), \Otil\left(\frac{m}{k}\right)\right)\right)
\end{align*}
-time, where $S$ is the total VWF size in the problem instance.

\subsection{Numerical Routine: Proximal Gradient Methods}

Proximal Gradient Methods is a family of iterative methods for minimizing objectives of the form $\min_{\bx}g(\bx) + h(\bx)$, where $g$ is smooth and convex.
These methods solves the problem by computing a series of \emph{proximal mapping}.
Given some current solution $\bx$, the \emph{proximal mapping} of $\bx$, or $p_\bx$, is the optimal solution to the following problem:
\begin{align*}
    p_\bx \coloneqq \arg\min_{\by}g(\bx) + \grad g(\bx)^\top (\by - \bx) + \frac{1}{2}\norm{\by - \bx}^2 + h(\by).
\end{align*}
The proximal mapping may be simpler to compute in the sense that we replace $g$ by its quadratic approximation, which may has more structure.
Later, we will see that computing $p_\bx$ is another instance of Problem~\eqref{gen_L2_dual}.

In our setting, the function $g$ is always $1$-smooth and $1/k$-strongly convex for some number $k \ge 1$.
That is, the following holds for any $\bx, \by \in \dom(g)$:
\begin{align*}
    g(\bx) + \grad g(\bx)^\top (\by - \bx) + \frac{1}{2k}\norm{\by - \bx}^2 \le g(\by) \le g(\bx) + \grad g(\bx)^\top (\by - \bx) + \frac{1}{2}\norm{\by - \bx}^2.
\end{align*}
Accelerated versions of the proximal gradient method~\cite{Tseng08, BeckT09, Nesterov13} computes an $(1+\eps)$-approximated solution using $O(\sqrt{k}\log(1/\eps))$ proximal mapping computations.

Problem~\eqref{gen_L2_dual} can be solved by proximal gradient methods.
Define $g(\bx) = 0.5 \bx^\top \bL(G) \bx$ and $h(\bx) = \sum_{u}f_u(x_u) + I_{\ge \bb}(\bx)$ where $I_{\ge \bb}(\bx)$ is 0 if $\bx \ge \bb$ or $\infty$ otherwise.
It is identical to the unconstrained problem of minimizing $g(\bx) + h(\bx)$.
Let $H$ be a preconditioner of quality $k$.
The function $g$ is $1$-smooth and $1/k$-strongly convex with respect to the norm $\norm{\bx}_{H} \coloneqq \sqrt{\bx^\top \bL(H) \bx}$.
Given any current solution $\bx$, computing its proximal mapping is equivalent to solving the following problem:
\begin{align*}
    \min_{\by \ge 0} \frac{1}{2}\by^\top \bL(H)\by + (\bL(G)\bx - \bL(H)\bx)^\top \by + \sum_u f_u(y_u).
\end{align*}
It is also an instance of Problem~\eqref{gen_L2_dual} since the family of VWFs is closed under adding a linear term.
Furthermore, the problem is defined on a preconditioner $H$ and can be solved recursively.

Past results assume exact computation of the proximal mapping.
In our case, computing proximal mapping is another instance of Problem~\eqref{gen_L2_dual} and can be solved only approximately.
In \cref{sec:ProxAGD}, we show similar convergence results under inexact proximal mapping computations.
That is, given an instance of Problem~\eqref{gen_L2_dual} and a quality $k$ preconditioner $H$, the instance can be solved in high-accuracy by reducing it to $\Otil(\sqrt{k})$ instances whose underlying graphs are all $H$.
Each of the $\Otil(\sqrt{k})$ instances can be solved up to $(1+1/\poly(n))$-approximation.
Alternatively, if $H$ can be constructed in $\Otil(m)$-time, the reduction yields the following bound on the run-time of our algorithm
\begin{align*}
    T(m, S) = \Otil\left(\sqrt{k}\cdot (m + S) + \sqrt{k}\cdot T(m_H, S + n)\right),
\end{align*}
where $m_H$ is the number of edges in the preconditioner $H$.

The fact that the preconditioner $H$ is a $\Otil(m/k)$-tree further bounds the run-time by
\begin{align*}
    T(m, S) 
    &= \Otil\left(m + \sqrt{k}\left(m + S + T(m_H, S + n)\right)\right) \\
    &= \Otil\left(\sqrt{k}\cdot(m + S) + \sqrt{k} \cdot T(\Otil(m/k), \Otil(m/k))\right).
\end{align*}
The original $\ell_2$-norm flow diffusion problem is an instance of Problem~\eqref{gen_L2_dual} with total VWF size $n$.
Our algorithm solves Problem~\eqref{L2_dual} in $T(m, n) = \Otil(m)$-time by setting $k=\poly\log n$.

\section{Preliminaries}
\label{sec:prelim}


\paragraph{General Notations}
In this work, we use bold fonts for vectors and matrices.
Given a vector $\bx \in \Real^d$, and a subset of indices $I \subseteq [d]$, we use $\bx[I]$ to denote the sub-vector of $\bx$ indexed by $I$.
We also write $\Diag(\bx)$ for the $d \times d$ diagonal matrix with $\Diag(\bx)_{i, i} = \bx_i$.
Given that $\bx$ is clearly a vector in the context, we write $\bX$ to denote $\Diag(\bx)$.
$\mb{0}$ and $\mb{1}$ denote the all zero and all one vectors.
$\mb{1}_u$ is the vector of all zero except the coordinate indexed by $u$ being one.

Given any finite set $S$, we denote $\Real^S$ as the set of $|S|$-dimensional vectors indexed by elements of $S$.
For any 2 finite sets $A, B$, we denote $\Real^{A \times B}$ as the set of $|A| \times |B|$ matrices with row indexed by $A$ and column indexed by $B$.
Given any set $S$, define $\chi_S(x)$ be the indicator function of $S$ with $\chi_S(x) = 1$ if $x \in S$ and $0$ otherwise.

Given 2 symmetric matrices $\bA, \bB$, we write $\bA \preceq \bB$ if $\bB - \bA$ is positive semi-definite (psd).

\paragraph{Norm}
In this work, we abuse the term norm to denote \emph{seminorm}.
Given some vector space $X$ over $\Real$, a mapping $p:X \to \Real$ is a \emph{norm} if (1) $p(\bx + \by) \le p(\bx) + p(\by)$ and (2) $p(s\bx) = |s|p(\bx)$.
Contrast to the common definition of the norm, non-zero vector can have norm $0$.
Given any psd matrix $\bA$, the norm induced by $\bA$ is defined as $\norm{\bx}_\bA \coloneqq \sqrt{\bx^\top \bA \bx}.$

\paragraph{Graph}
In this work, we consider only undirected, positive-weighted multi-graphs.
When the graph $G$ is cleared from the context, we use $n$ to denote the number of vertices and $m$ being the number of edges.
For an edge $e$, its conductance is denoted by $c(e)$, and its resistance, $r(e) = 1/c(e).$
For any edge $e$ with endpoints $u, v$ and conductance $c$, we orient $e$ arbitrarily and denote it as $e=c(u, v).$
Without loss of generality, we assume $U$, the ratio between largest and smallest edge conductance, is bounded by $\poly(n).$

Given a graph $G=(V, E, c)$, $\bB(G) \in \Real^{E \times V}$ is the edge-vertex incidence matrix of $G$, i.e, $\bB(G)_{c(u, v)} = \mb{1}_u - \mb{1}_v.$
Let $\bC(G)$ be the diagonal matrix in $\Real^{E \times E}$ with $\bC(G)_{e, e}=c(e)$.
The \emph{Laplacian matrix} of $G$ is defined as $\bL(G) \coloneqq \bB(G)^\top \bC(G) \bB(G).$

A graph $H$ is a \emph{$\kappa$-Spectral Sparsifier} of $G$ if both have the same vertex set and similar Laplacian, i.e $\bL(G) \preceq \bL(H) \preceq \kappa \bL(G).$
If $\kappa = \Omega(1)$, such object can be computed efficiently:
\begin{fact}[Fast Construction of Spectral Sparsifiers, \cite{KoutisLP15}]
\label{fact:SpecSparsifier}
One can compute $H$ which is a subgraph of $G$ with $O(n \log n)$ edges in $O(m\log^2 n\log\log n)$-time.
With high probability, $H$ is a 2-spectral sparsifier of $G$, i.e. $\bL(G) \preceq \bL(H) \preceq 2 \bL(G)$.
\end{fact}

The construction of spectral sparsifier is the \emph{only} randomized component in this paper.

\paragraph{Flow and Potential}
A \emph{flow} $\bf\in\Real^E$ is a vector indexed by the edge set $E$.
For any edge $e=c(u, v)$, the magnitude of $f_e$ is the amount of mass crossing the edge $e.$
If $f_e > 0$, mass is sent from $u$ to $v$ and vice versa.
The vector $\bB(G)^\top \bf \in \Real^V$ is the \emph{residue} of $\bf.$

We use the term \emph{potential} to denote vectors indexed by the vertex set $V$.
The flow induced by a potential $\bx \in \Real^V$, or the \emph{potential flow of $\bx$}, is defined as the flow $-\bC(G)\bB(G)\bx.$


\subsection{Vertex Weighting Functions (VWF)}

We present the formal definition of \emph{Vertex Weighting Functions (VWF)}.


\begin{definition}[Vertex Weighting Function]
\label{defn:VtxWeightFunction}
A continuous convex function $f$ whose domain is $[L, \infty), L \le 0$ is a \emph{Vertex Weighting Function} (VWF) of size $k$ if
\begin{enumerate}
    \item $f(0) \le 0$,
    \item $f(x)$ is piece-wise quadratic with split points $(L = s_0), s_1, \ldots, s_k = \infty$,
    \item $f'(x)$ is continuous, concave, and constant for $x \ge s_{k-1}$.
\end{enumerate}
We abuse the notation and use $\Abs{f}$ to denote the size of $f$.
\end{definition}

We use a list of tuples $\{(s_i, r_i, a_i, b_i)\}_{i=0}^k$ to denote a VWF $f$ of size $k$ where $f(x) = (1/2)r_ix^2 + a_ix + b_i, x \in [s_i, s_{i+1}).$
Here are some obvious properties about VWFs.

\begin{proposition}
For any $a \in \Real$, $f(x) = ax$ is a VWF of size $1$.
\end{proposition}

\begin{proposition}
Given 2 VWF $f$ and $g$, $f + g$ is also a VWF of size at most $|f|+|g|$.
\end{proposition}

\begin{proposition}
Given a VWF $f$ and some $x \in \dom(f)$, the re-oriented version $\Bar{f}(t) = f(x+t) - f(x)$ is a VWF of size $|f|$.
\end{proposition}




\subsection{Generalized Diffusion Problems}
\label{sec:generalized_diffusion}
Our algorithm is built upon layers of recursion.
To clearly analyze the algorithm, a formal definition of problem instance helps.

\begin{definition}[Diffusion Instance]
\label{defn:DiffusionInstance}
A \emph{Generalized Diffusion Instance} (or diffusion instance) is a tuple $\GG$,
\[
\GG \coloneqq (G, \bf^\GG, \bb^\GG),
\]
where $G=(V, E, c)$ is a graph, $\bf^\GG$ is the set of VWFs on each vertex of $V^\GG$ and potential lower bounds $\bb^\GG \in \Real^{V^\GG}_{\le 0}$.
The size of $\GG$, denoted by $|\GG|$, is the total size of each $f_u$ plus the number of edges.

For any vertex potential $\bx \in \Real^{V}$, $\bx$ is feasible for $\GG$ if $\bx \ge \bb^\GG$.
\end{definition}

Given an instance, we can define the diffusion problem induced by the instance.

\begin{definition}[Objective, $\EE^\GG$]
\label{defn:Objective}
Given a diffusion instance $\GG=(G, \bf, \bb)$, and a feasible potential $\bx$ for $\GG$, the associated objective function, or the energy, of $\bx$ is given as
\[
    \EE^\GG(\bx) \coloneqq \frac{1}{2} \T{\bx}\bL(G){\bx} + \sum_{u \in V} f_u(x_u).
\]

The \emph{diffusion problem on $\GG$} is to find a feasible potential $\bx^*$ for $\GG$ minimizing the energy, i.e.,
\[
    \EE(\GG) \coloneqq \min_{\bx \ge \bb} \EE^\GG(\bx).
\]
\end{definition}


Note that $\mb{0}$ is always a feasible potential of zero energy, the minimum energy of every diffusion instance is non-positive.
Without loss of generality, we assume every feasible potential has non-positive energy.

Problem~\ref{L2_dual} is equivalent to the diffusion problem on the instance
\[
\GG^0 = (G, \{f_u(x) = d_u x\}_{u \in V}, \mb{0}).
\]

Instead of finding exact minimizers, we are interested in finding feasible potential approximating optimal potential.
Here we present the notion of approximation used in this work:
\begin{definition}
\label{defn:DiffusionComplexity}
Given a diffusion instance $\GG$ and an error parameter $\kappa \ge 1$, a feasible potential $\bx$ is a \emph{$\kappa$-approximated optimal potential} to $\GG$ if
\begin{align*}
    \EE^\GG(\bx) \le \frac{1}{\kappa}\EE(\GG).
\end{align*}

$T(m, S, \kappa)$ is the time to compute a $\kappa$-approximated optimal potential to $\GG$ with probability at least $1-n^{-10}.$ where $m$ is the number of edges in $G$ and $S$ is the total size of VWFs $\bf^\GG.$
\end{definition}

Our recursive algorithm reduces the diffusion problem on some instance to ones on some other instances.
The reduction is formalized using the notion of embeddability.


\begin{definition}[$\HH \preceq_\kappa \GG$]
\label{defn:embed}
Given 2 diffusion instances $\HH, \GG$, we write $\HH \preceq_\kappa \GG$ for some $\kappa \ge 1$ if there's a mapping $\MM_{\HH \to \GG}:\Real^{V^\HH} \to \Real^{V^\GG}$ with the following properties:
\begin{enumerate}
    \item Given any feasible potential $\bx^\HH$ for $\HH$, $\bx^\GG \coloneqq \MM_{\HH \to \GG}(\bx^\HH)$ is feasible for $\GG$, and
    \item $\EE^\HH(\bx^\HH) \ge \kappa \EE^\GG\left(\frac{1}{\kappa}\bx^\GG\right)$ holds.
\end{enumerate}
Notice that $\kappa^{-1} \bx^\GG$ is also feasible for $\GG$ since $\kappa \ge 1.$
\end{definition}

Here we present some basic facts on embeddability.

\begin{fact}
$\GG \preceq_1 \GG$ with identity mapping.
\end{fact}
\begin{fact}
\label{fact:embedTrans}
If $\GG_1 \preceq_{\kappa_1} \GG_2$ and $\GG_2 \preceq_{\kappa_2} \GG_3$, we have $\GG_1 \preceq_{\kappa_1 \kappa_2} \GG_3$ with $\MM_{\GG_1 \to \GG_3} = \MM_{\GG_2 \to \GG_3} \circ \MM_{\GG_1 \to \GG_2}$.
\end{fact}

Embeddability also preserves approximated solutions:
\begin{fact}
\label{fact:embedApprox}
If $\HH \preceq_\kappa \GG \preceq_1 \HH$ and $\bx$ is an $\alpha$-approximated optimal potential for $\HH$, $\kappa^{-1}\MM_{\HH \to \GG}(\bx)$ is an $\kappa \alpha$-approximated optimal potential for $\GG$.
\end{fact}
\begin{proof}
Let $\bx^*_\GG$ be the optimal potential for $\GG$.
We have
\begin{align*}
    \EE(\GG)
    \underbrace{\ge}_{\GG \preceq_1 \HH} \EE^\HH\left(\MM_{\GG \to \HH} \bx^*_\GG\right)
    \ge \EE(\HH)
    \ge \alpha \EE^\HH(\bx)
    \underbrace{\ge}_{\HH \preceq_\kappa \GG} \alpha \kappa \EE^\GG\left(\frac{1}{\kappa}\MM_{\HH \to \GG}(\bx)\right)
\end{align*}
\end{proof}


Given a feasible potential $\bx$ for some instance $\GG$, $\bx+\Delta^*$ is an optimal solution for $\GG$, where $\Delta^*$ is a minimizer of the following problem:
\[
    \min_{\Delta + \bx \ge \bb^\GG} \EE^\GG(\bx + \Delta) - \EE^\GG(\bx).
\]
Plugging in the definition, we have
\begin{align*}
    \EE^\GG(\bx + \Delta) - \EE^\GG(\bx)
    &= \frac{1}{2}\T{\Delta}\bL(G)\Delta + (\bL(G)\bx)^\top\Delta + \sum_{u \in V^\GG} (f_u(x_u + \Delta_u) - f_u(x_u)) \\
    &= \frac{1}{2}\T{\Delta}\bL_\GG\Delta + \sum_{u \in V^\GG} g_u(\Delta_u),
\end{align*}
where $g_u(y) \coloneqq (\bL(G)\bx)_u y + f_u(x_u + y) - f_u(x_u), \forall u \in V^\GG.$
By basic properties of VWF, $g_u$ is a VWF of size $|f_u|$ and can be constructed in time $O(|f_u|).$

\begin{definition}[Residual Instances and Problems]
\label{defn:ResidualProb}
Given a diffusion instance $\GG$ and some feasible potential $\bx$, the \emph{residual instance of $\GG$ at $\bx$} is defined as
\[
    \RR(\GG; \bx) \coloneqq (G, \{g_u\}_{u \in V^\GG}, \bb^\GG - \bx).
\]
The \emph{residual problem on $\GG$ at $\bx$} is the diffusion problem on instance $\RR(\GG; \bx)$.
We use $\RR$ to denote the residual instance if both $\GG$ and $\bx$ are clear from the context.
\end{definition}

\begin{fact}
\label{fact:ResidualConstruction}
The residual instance $\RR(\GG; \bx)$ has same size as the original instance $\GG$.
In addition, it can be constructed in $O(|\GG|)$ time.
\end{fact}

\subsection{Numerical Issue}
Discussion on numerical stability is beyond the scope of the paper.
The following assumption is made throughout the execution of the algorithm:
\begin{assumption}
\label{assump:polyNum}
In this paper, every number encountered has its absolute value being either zero or within the range $[n^{-c}, n^c]$ for some universal constant $c > 0$.
\end{assumption}

\section{The Main Algorithm: Proving Theorem~\ref{thm:L2Main}}

In this section, we prove Theorem~\ref{thm:L2Main} via Iterative Refinement~\cite{AKPS19}.
Roughly speaking, we make sequential calls to an oracle that can solve diffusion problems approximately and boost solutions given by the oracle to form a highly accurate solution.
Given some feasible vertex potential $\bx$, we use the oracle to solve the residual problem approximately.
The approximated solution is a direction of small residual energy.
Moving toward that direction gives certain amount of decrease in the energy.

The following lemma formalizes the oracle for proving \cref{thm:L2Main}.

\begin{lemma}
\label{lemma:RecursiveApproxDiffusion}
There exists a randomized algorithm (\cref{algo:RecursiveApproxDiffusion}) invoked with the syntax
\[
    \bx = \textsc{RecursiveApproxDiffusion}(\GG)
\]
that on input of a diffusion instance $\GG$ of $m$ edges and $O(m)$-sized VWFs, outputs a $2$-approximated optimal potential for $\GG$ with high probability such that the algorithm runs in $O(m \log^8 n)$-time.

That is, $T\left(m, O(m), 2\right) = O(m \log^8 n)$ holds.
\end{lemma}

The following lemma formalizes the iterative refinement framework.

\begin{lemma}
\label{lemma:IterRefine}
There exists an algorithm (\cref{algo:Iter}) invoked with the syntax
\[
    \bx = \textsc{Iter}(\GG, \eps, \OO, \alpha)
\]
that on input of a diffusion instance $\GG$ of $m$ edges and $O(S)$-sized VWFs, an error parameter $\eps > 0$, and an oracle $\OO$ that can output a $\alpha$-approximate optimal potential, outputs a $1 + \eps$-approximated optimal potential for $\GG$ such that the algorithm runs in
\[
O\left(\alpha \log\frac{1}{\eps}\left(m + T\left(m, O(S), \alpha\right)\right)\right)
\]
time by making $O(\alpha \log(1/\eps))$ oracle calls to $\OO$ on diffusion instance based on graph $G$.
\end{lemma}
\begin{proof}
The algorithm is presented as \cref{algo:Iter}.

Given any feasible potential $\bx$, the minimum energy of its residual problem is $\EE(\RR) = \EE(\GG) - \EE^\GG(\bx)$.
The fact and guarantee of the oracle $\OO$ yield that for every step $i$,
\begin{align*}
    \EE^\GG(\bx^{i+1}) - \EE(\GG)
    &= \EE^\GG(\bx^{i}) - \EE(\GG) + \EE^\GG(\bx^{i+1}) - \EE(\bx^{i}) \\
    &= \EE^\GG(\bx^{i}) - \EE(\GG) + \EE^{\RR^i}(\Delta^i) \\
    &\le \EE^\GG(\bx^{i}) - \EE(\GG) + \frac{1}{\alpha}\left(\EE(\GG) - \EE^\GG(\bx^i)\right) \\
    &= \left(1 - \frac{1}{\alpha}\right)\left(\EE^\GG(\bx^{i}) - \EE(\GG)\right),
\end{align*}
where the second equality comes from that $\bx^{i+1}=\bx^{i} + \Delta^i$.

Induction yields that
\begin{align*}
    \EE^\GG(\bx^{T}) - \EE(\GG) 
    \le \left(1 - \frac{1}{\alpha}\right)^T\left(\EE^\GG(\bx^{0}) - \EE(\GG)\right) \le -\eps \EE(\GG),
\end{align*}
where the second inequality comes from that $\EE^\GG(\bx^{0}) \le 0$ and $1-x \le \exp(-x)$ for any $x$.


\end{proof}

\begin{algorithm}
\caption{Solve Diffusion Dual Problem}
\label{algo:Main}
\begin{algorithmic}[1]
    \Procedure{L2Diffusion}{$G=(V,E,c), \bd, \eps$}
        \State $\bx = \textsc{Iter}(\GG^0, \eps, \textsc{RecursiveApproxDiffusion}, 2)$
        \State \Return $\bx$
    \EndProcedure
\end{algorithmic}
\end{algorithm}

\begin{algorithm}
\caption{$\ell_2$ Iterative Refinement}
\label{algo:Iter}
\begin{algorithmic}[1]
    \Procedure{Iter}{$\GG, \eps, \OO, \alpha$}
        \State $T \coloneqq \lceil \alpha \log \frac{1}{\varepsilon} \rceil$.
        \State $\bx^0 = \mb0$
        \For{$i = 0, 1, \ldots, T-1$}
            \State Construct the residual problem $\RR^i = \RR(\GG; \bx^i)$ via \cref{fact:ResidualConstruction}
            \State $\Delta^i \coloneqq \OO(\RR^i)$
            \State $\bx^{i+1} \coloneqq \bx^i + \Delta^i$
        \EndFor
        \State \Return $\bx_T$.
    \EndProcedure
\end{algorithmic}
\end{algorithm}

\begin{proof}[Proof of Theorem~\ref{thm:L2Main}]
Straightforward from \cref{lemma:RecursiveApproxDiffusion} and \cref{lemma:IterRefine}.
\end{proof}


\subsection{Prove Lemma~\ref{lemma:RecursiveApproxDiffusion}}

In this section, we state technical lemmas to prove \cref{lemma:RecursiveApproxDiffusion}.

\subsubsection*{Inexact Proximal Accelerated Gradient Descent}

Accelerated proximal methods are used to solve a family of lightly structured convex problems: minimizing a composite objective $g(\bx) + h(\bx)$ where $h$ is convex and $g$ is smooth with respect to a certain norm $\norm{\cdot}.$
Using proximal methods, one can reduce the minimization problem to a small number of simpler problems: finding $\Delta$ minimizing $\grad f(\bx)^\top \Delta + (1/2)\norm{\Delta}^2 + h(\bx + \Delta)$.
The sub-problem is called \emph{Proximal Mapping}.

There is a long line of research in accelerated proximal gradient methods, for example \cite{Tseng08, BeckT09, Nesterov13}.
All these works assume exact proximal mapping computation.
The assumption does not hold in our setting since computing proximal mapping is equivalent to solving another diffusion problem.

\cite{SRB11} proves that the acceleration still holds given inexact proximal mapping.
The quality of the inexact proximal mapping oracle affects the iteration complexity.
However, we cannot directly use their results since they assume smoothness in only Euclidean norm.
We will follow their idea and adapt the proof for general norms.

The result is formalized as follow and will be proved in \cref{sec:ProxAGD}:
\begin{lemma}
\label{lemma:ProxAGD}
There exists an algorithm (\cref{algo:ProxAGD}) invoked with the syntax
\[
    \bx = \textsc{ProxAGD}(\GG, H, \kappa, \OO)
\]
that on input of a diffusion instance $\GG$ of $m$ edges and $O(S)$-sized VWFs, a graph $H$ which is a $\kappa$-spectral sparsifer of $G$ with $\kappa \in [1, m]$, and an oracle $\OO$ that can output $1+n^{-c}$-approximate optimal potential for arbitrary large constant $c > 0$, outputs a $2$-approximated optimal potential for $\GG$ such that
the algorithm runs in
\[
O\left(\sqrt{\kappa}\left(m + T\left(|E(H)|, O(S), 1+n^{-c}\right)\right)\right)
\]
time by making $O(\sqrt{\kappa})$ oracle calls to $\OO$ on diffusion instance based on graph $H$.
\end{lemma}

\subsubsection*{$J$-Trees and Vertex Elimination}

\cref{lemma:ProxAGD} gives a recursive scheme for designing diffusion solvers.
However, one have to construct the $\kappa$-spectral sparsifier $H$ that can provide enough decrease in the problem size.
Therefore, the resulting algorithm can be efficient.

\cite{Madry10} introduces the notion of \emph{$J$-Tree}, which is a union of $j$ rooted trees and a graph induced by the $j$ roots.
Solving diffusion problem on a $j$-tree can be reduced to solving diffusion problem on a graph of $j$ vertices.
Also, one can construct $j$-tree as a spectral sparsifier.

These ideas are stated as following lemmas:
\begin{lemma}
\label{lemma:JTreeSparsify}
There exists an algorithm (\cref{algo:JTreeSparsify}) invoked with the syntax
\[
    H = \textsc{JTreeSparsify}(G, j)
\]
that on input of a graph $G=(V, E, c)$ of weight ratio $U = \poly(n)$, positive integer $j \ge 10$, outputs a $j$-tree $H=(V, E_H, c_H)$ such that
\begin{enumerate}
    \item $H$ is a $(C m \log n \log\log n/j)$-spectral sparsifier of $G$ for some universal constant $C > 0$,
    \item the core graph of $H$ has $O(j\log n)$ edges,
    \item the algorithm runs in time $O(m \log^2 n \log\log n)$.
\end{enumerate}
\end{lemma}

\begin{lemma}
\label{lemma:JTreeSolve}
There exists an algorithm (\cref{algo:JTreeSolve}) invoked with the syntax
\[
    \bx = \textsc{JTreeSolve}(\HH, \eps, \OO)
\]
that on input of a diffusion instance $\HH$ whose underlying graph $H$ is a $j$-tree with $O(j \log n)$ core-edges and VWFs have total size $O(m)$, and an oracle $\OO$ that can output $2$-approximate optimal potential, outputs a $1 + \eps$-approximated optimal potential $\bx$ for $\HH$ such that the algorithm runs in time
\begin{align*}
    O\left(m \log^2 n + \log \left(\frac{1}{\eps}\right)\left(m + T\left(j \log n, j \log n, 2\right)\right) \right).
\end{align*}
\end{lemma}

\subsubsection*{Prove \cref{lemma:RecursiveApproxDiffusion}}

\begin{algorithm}
\caption{Solve Diffusion Problem Approximately}
\label{algo:RecursiveApproxDiffusion}
\begin{algorithmic}[1]
    \Procedure{RecursiveApproxDiffusion}{$\GG$}
        \If{$|E(G)| = O(1)$}
            \State Solve the diffusion directly in $O(1)$-time.
        \Else
            \State $\kappa \coloneqq \log^{8} n$
            \State $j \coloneqq C m \log n \log\log n / \kappa = O(m \log^2 n / \kappa)$
            \State $H = \textsc{JTreeSparsify}(G, j)$
            \State $\bx = \textsc{ProxAGD}(\GG, H, \kappa, \textsc{JTreeSolve}(\cdot,n^{-c},\textsc{RecursiveApproxDiffusion}))$
            \State \Return $\bx$
        \EndIf
    \EndProcedure
\end{algorithmic}
\end{algorithm}

\begin{proof}[Proof of \cref{lemma:RecursiveApproxDiffusion}]
The correctness is obvious from guarantees of \cref{lemma:ProxAGD}, \cref{lemma:JTreeSparsify}, and \cref{lemma:JTreeSolve}.

The time complexity can be bound via recursion.
Recall the notation that $m$ is the number of edges in $G$, $n$ the number of vertices, and $S=O(m)$ the total size of VWFs of $\GG$.
When the input has constant size, we have $T(O(1), O(1), 2) = O(1).$
Otherwise, we can write down the following recursion equations:
\begin{align*}
    T(m, O(m), 2) &= O\left(\underbrace{m\log^3 n}_{\textsc{JTreeSparsify}} + \underbrace{\sqrt{\kappa} \left[m + T\left(|E(H)|, O(m), 1+n^{-c}\right)\right]}_\textsc{ProxAGD}\right)
\end{align*}
Since $H$ is a $O(m \log^2 n / \kappa)$-tree with $O(m \log^3 n / \kappa)$ core-edges, \cref{lemma:JTreeSolve} yields that
\begin{align*}
    T\left(|E(H)|, O(m), 1+n^{-c}\right) 
    = O\left(m\log^2 n +\log n \cdot T\left(O\left(\frac{m\log^3 n}{\kappa}\right), O\left(\frac{m\log^3 n}{\kappa}\right), 2\right)\right).
\end{align*}
Combining these bounds and the definition of $\kappa = \log^{8} n$ yields
\begin{align*}
    T(m, O(m), 2) 
    &= O\left(m\log^3 n + \sqrt{\kappa} \left[m + T\left(|E(H)|, O(m), 1+n^{-c}\right)\right]\right)\\
    &= O\left(m\log^3 n + \sqrt{\kappa} \left[m\log^2 n + \log n \cdot T\left(O\left(\frac{m\log^3 n}{\kappa}\right), O\left(\frac{m\log^3 n}{\kappa}\right), 2\right)\right]\right) \\
    &= O\left(m\log^6 n + \log^5 n \cdot T\left(O\left(\frac{m}{\log^5 n}\right), O\left(\frac{m}{\log^5 n}\right), 2\right)\right) = O(m \log^7 n).
\end{align*}
\end{proof}

\section{J-Tree Spectral Sparsifier}
\label{sec:JTreeSparsify}

In this section, we prove Lemma~\ref{lemma:JTreeSparsify}.
The idea of constructing spectral sparsifier is similar to one of \cite{SpielmanT03} and \cite{KPSW19}.
First, a low stretch spanning tree is computed and we route every edge via some tree path.
Then we shortcut some routing paths so that no tree edge is highly congested.

\subsection{Preliminaries on J-Trees}

\paragraph{Trees and Forests}
Given a forest $F$, for any 2 connected vertices $u$ and $v$, the unique $uv$-path in $F$ is denoted as $F[u, v]$.
Given an edge $e=c(u, v)$ not necessary inside the forest $F$, $F(e)$ is used to denote $T[u, v].$
When the forest $F$ is rooted, $\rho_F(u)$ is defined as the root of the corresponding tree for any vertex $u$.
$F(u)$ is used to denote the component of $F$ containing $u$.

Given a spanning tree $T$ and an arbitrary edge $e=c(u, v)$, the \emph{stretch} of $e$ on $T$ is defined as $\str_T(e) \coloneqq r(T(e)) / r(e)$, where $r$ is the resistance of edge.
For any tree edge $\Bar{e} \in T$, the \emph{weighted congestion} of $\Bar{e}$ in $T$ is defined as $\Cong_T(\Bar{e}) \coloneqq \sum_{e: \Bar{e} \in T(e)} c(e).$

It is known that a spanning tree with low total stretch can be computed in near linear time.
The result is formalized as follows:
\begin{lemma}[\cite{AbrahamN12:journal}]
\label{lemma:LSST}
Given a graph $G=(V,E,c)$, one can compute a spanning tree $T$ with
\[
    \sum_{e \in G}\str_T(e) \le c_{AN} m \log n (\log\log n)
\]
in $O(m \log n \log\log n)$-time for some universal constant $c_{AN}$.
\end{lemma}

\paragraph{Flow Routing}
Given 2 graphs $G, H$ of same vertex set, a \emph{Flow Routing of $G$ in $H$} is a linear mapping $\Pi$ from $G$'s flow space to $H$'s preserving residue, i.e., for any flow $\bf$ of $G$, $\bB(G)^\top \bf = \bB(H)^\top \Pi(\bf).$

In this work, all flow routings are path-based.
That is, given a flow routing $\Pi$ of $G$ in $H$, for any edge $e$ in $G$, the flow of sending one unit across $e$ is has its image under $\Pi$ in $H$ being unit flow along a path.
In this case, we use $\Pi(e)$ to denote the path in $H$.
Clearly, when $e=c(u, v)$, $\Pi(e)$ is a $uv$-path in $H$.
The following equality regard path-based flow routing is useful:
\begin{fact}
$\Pi(\bf)_{\Bar{e}} = \sum_{e: \Bar{e} \in \Pi(e)} f_e$.
\end{fact}

\paragraph{J-Tree\cite{Madry10}}
A graph $H$ is a \emph{$j$-tree} if one can partition $H$ into 2 edge-disjoint subgraphs $C$ and $F$ where $C$ has $j$ vertices and $F$ is an union of $j$ trees each of which has exactly one vertex in $C$.
The subgraph $C$ is the \emph{core} of $H$ and $F$ is the \emph{envelope} of $H$.

\subsection{Spectral Approximation via Flow Routing}

The $j$-tree spectral sparsifiers will not be a subgraph of the original graph.
Previous constructions of ultra-sparsifiers use Support Theory~\cite{SpielmanT03} to prove the spectral similarity.
This idea is not applicable in our setting.
Motivated by the preconditioner construction in \cite{KPSW19}, we relate the spectral similarity with the quality of the flow routing.
It is formalized as the following lemma.

\begin{lemma}
\label{lemma:FlowRoutingToSpectralApprox}
Given 2 graphs $G, H$ on the same vertex set, flow routings $\Pi_{G \to H}, \Pi_{H \to G}$ such that
\begin{enumerate}
    \item for any flow $\bf$ on $G$, $\norm{\Pi_{G \to H}(\bf)}_{\bR(H)}^2 \le 10\norm{\bf}_{\bR(G)}^2$ holds, and
    \item for any flow $\by$ on $H$, $\norm{\Pi_{H \to G}(\by)}_{\bR(G)}^2 \le 10\norm{\by}_{\bR(H)}^2$ holds as well,
\end{enumerate}
we have that $0.1 \bL(G) \preceq \bL(H) \preceq 10 \bL(G).$
\end{lemma}
\begin{proof}
We will instead prove the following:
\[
    0.1 \pinv{\bL(G)} \preceq \pinv{\bL(H)} \preceq 10 \pinv{\bL(G)}.
\]
The inequality is valid since both $H$ and $G$ are connected graphs, thus has the same null space $span(\mb1)$.

Given any demand $\bd$ on vertices such that $\mb{1}^\top \bd=0$, the value $\bd^\top \pinv{\bL(G)} \bd$ is the energy of the \emph{electrical flow} that routes $\bd$.
That is, there is some flow $\bf$ such that $\bB(G)^\top \bf = \bd$ and $\bd^\top \pinv{\bL(G)} \bd = \bf^\top \bR(G) \bf$.
Let $\Bar{\bf} = \Pi_{G \to H} \bf$, we know $\Bar{\bf}$ routes $\bd$ in $H$, i.e. $\bd = \bB(H)^\top \Bar{\bf}$.
Since electrical flow has minimum energy among flows of same residue, the following holds
\begin{align*}
    \bd^\top \pinv{\bL(G)} \bd = \bf^\top \bR(G) \bf \ge 0.1 \Bar{\bf}^\top \bR(H) \Bar{\bf} \ge 0.1 \bd^\top \pinv{\bL(H)} \bd,
\end{align*}
where the first inequality comes from Condition 1, and the second inequality uses the optimality of electrical flow.
So, $\pinv{\bL(H)} \preceq 10 \pinv{\bL(G)}$ holds.

Similarly, we have $\pinv{\bL(G)} \preceq 10 \pinv{\bL(H)}$.
\end{proof}

\subsection{Flow Routing on J Trees}

In this section, we introduce the notion of $j$-tree formally and present criterias for $j$-tree being a spectral sparsifier.
Also, we describe a canonical way of routing flows in a $j$-tree.

\begin{definition}[Forest Routing]
Given a graph $G=(V, E, c)$, and a spanning rooted forest $F$, for every edge $e = c(u, v)\in E$, we define its \emph{Forest Path} $P_F(e)$ and $e$'s image under \emph{Forest Edge Moving}, $\Hat{e} = \mathrm{Move}_F(e)$ as follows:
\begin{enumerate}
    \item[1.](Tree edge) If $e \in F$. In this case, $P_F(e) = \mathrm{Move}_F(e) = e.$

    \item[2.](Self-Loop) If $\rho_F(u) = \rho_F(v)$.
    In this case, $P_F(e)$ is the unique $uv$-path in $F$ ($F[u, v]$) and $\mathrm{Move}_F(e)$ is null.
    
    \item[3.](Core-Edge) If $\rho(u) \neq \rho(v)$.
    In this case, $\mathrm{Move}_F(e)$ is a distinct edge $\rho(u)\rho(v)$.
    The path $P_F(e)$ is the tree path $F[u, \rho(u)]$ followed by the new edge $\rho(u)\rho(v)$ and then the tree path $F[\rho(v), v].$
    
    The conductance of $\mathrm{Move}_F(e)$ is set to be $c$.
\end{enumerate}
\end{definition}

\begin{definition}[Canonical J Tree]
\label{defn:CanonicalJTree}
Given a graph $G=(V, E, c)$, a spanning rooted forest $F$,
the \emph{canonical $j$-tree} of $G$ with respect to $F$, denoted by $G_F$, is defined as follows:
\begin{enumerate}
    \item $V(G_F) = V$,
    \item For every edge $e=c(u, v) \in E$, add edge $\mathrm{Move}_F(e)$ to $G_F$ is it is not null.
    When multi-edges appear in $G_F$, we compress them into one edge with conductance being the sum of each individual conductance.
    \item $F$ is denoted as the \emph{envelope} of $G_F$.
    \item $G_F[C]$ is denoted as the \emph{core} of $G_F$, where $C$ is the set of roots of the rooted forest $F$.
\end{enumerate}
\end{definition}

\begin{definition}[Flow Routing in Canonical J Tree]
The \emph{canonical flow routing} of $G$ into $G_F$ is the path-based flow routing $\Pi_{G \to G_F}$ that maps each edge $e$ of $G$ to $P_F(e)$, the forest path of $e$.
When the graphs $G$ and $G_F$ are cleared from the context, the subscript is often ignored and we use only $\Pi$ to denote the flow routing.
We can also route $G_F$ in $G$ via reverse routing.
\end{definition}

\subsection{Prove \cref{lemma:JTreeSparsify} via Low-Stretch Spanning Forests}
\label{sec:proveJTreeSparsify}

In this section, we introduce the notion of \emph{Low-Stretch Spanning Forests} and use this to prove \cref{lemma:JTreeSparsify}.

Given a graph $G$ and a rooted spanning forest $F$, we define the \emph{Local Stretch} of each edge $e$ with respect to each component of $F$.
That is,
\begin{definition}[Local Stretch]
\label{defn:localStretch}
Given a graph $G=(V, E, c)$, a rooted spanning forest $F$, and the canonical routing of $G$ into $G_F$, for every edge $e=c(u, v) \in E$ and every component $C$ of $F$, the \emph{Local Stretch} of $e$ in $C$ is defined as $\str_F(e; C) \coloneqq r(P_F(e) \cap C) / r(e)$.

For every component $C$ of $F$, we define the \emph{Total Local Stretch} of $C$ as $\str_F(C) = \sum_{e \in E} \str_F(e; C)$.
Furthermore, we define the \emph{Maximum Local Stretch} of $F$ as $\str(F) = \max_{C} \str_F(C)$.
We say $F$ is a \emph{Low-Stretch Spanning Forest} if $\str(F)$ is small.
\end{definition}

The following lemma relates the spectral similarity between $G$ and $G_F$ with $\str(F)$.
The proof is deferred to \cref{sec:proveJTreeRouting}.
\begin{lemma}
\label{lemma:JTreeRouting}
Given a graph $G=(V,E,c)$, a rooted spanning forest $F$ containing $j$ trees, we can construct a $j$-tree $H$ such that $\bL(G) \preceq \bL(H) \preceq 100\kappa \bL(H)$ if $\str(F) \le \kappa$.
\end{lemma}

To construct a $j$-tree sparsifier using \cref{lemma:JTreeRouting}, we need to compute a spanning forest $F$ with small $\str(F)$.
The construction computes a low-stretch spanning tree, then modify the spanning tree to a spanning forest.
It is formalized as the following lemma and will be proved in \cref{sec:FindForest}.
\begin{lemma}
\label{lemma:FindForest}
There exists an algorithm (\cref{algo:FindForest}), which we invoke with the syntax
\[
    (F, C) \coloneqq \textsc{FindForest}(G, j)
\]
that on input of a graph $G=(V,E,c)$, and a positive integer $j \ge 10$, outputs a rooted spanning forest $F$ with roots $C$ such that
\begin{enumerate}
    \item $|C| \le j$.
    \item $\str(F) \le \kappa \coloneqq 100 C_{AN} m\log n \log\log n /~ j$, and
    \item The algorithm runs in $O(m \log n\log\log n)$ time.
\end{enumerate}
\end{lemma}

\begin{algorithm}
\caption{Construct $j$-tree ultra sparsifier}
\label{algo:JTreeSparsify}
\begin{algorithmic}[1]
    \Procedure{JTreeSparsify}{$G=(V,E,c), j$}
        \State $(F, C) \coloneqq \textsc{FindForest}(G, j)$
        \State $\kappa \coloneqq 100 C_{AN} m\log n \log\log n / j$
        \State Let $H$ be the $j$-tree $100\kappa$-spectrally approximating $G$ via \cref{lemma:JTreeRouting}.
        \State Replace the core of $H$ by its 2-spectral sparsifier via \cref{fact:SpecSparsifier}.
        \State \Return $H$
    \EndProcedure
\end{algorithmic}
\end{algorithm}

\begin{proof}[Proof of \cref{lemma:JTreeSparsify}]
\cref{algo:JTreeSparsify} is the algorithm for \cref{lemma:JTreeSparsify}.

\textbf{Correctness:}
The rooted spanning forest $F$ satisfies conditions in \cref{lemma:JTreeRouting}.
Therefore, the $j$-tree $H$ constructed in line 4 is a $100\kappa$-spectral sparsifier of $G$.

By replacing the core with its 2-spectral sparsifier and properties of spectral sparsifiers, $H$ is a $200\kappa$-spectral sparsifier.

\textbf{Time Complexity:}
$\textsc{FindForest}$ runs in $O(m \log n\log\log n)$-time.
Construction in \cref{lemma:JTreeRouting} takes linear time.
Constructing 2-spectral sparsifier on the core of $H$ takes $O(m\log^2 n \log\log n)$-time.
Therefore, the total time complexity is $O(m\log^2 n \log\log n)$-time.
\end{proof}

\subsection{Prove \cref{lemma:JTreeRouting} via Scaling}
\label{sec:proveJTreeRouting}

In this section, we prove \cref{lemma:JTreeRouting} by bounding the quality of the canonical flow routing of $G$ into $H$, the $j$-tree with low local stretch.
Combining with \cref{lemma:FlowRoutingToSpectralApprox} yields the desired statement.

\begin{proof}[Proof of \cref{lemma:JTreeRouting}]
Consider the graph $\Bar{G} = (V, E, \Bar{c})$ with
\begin{align*}
    \overline{c}(e) = \begin{cases}
        \kappa c(e)&, e \in F\\
        c(e)&, e \in E \setminus F.
    \end{cases}
\end{align*}
We have that $\bL(G) \preceq \bL(\Bar{G}) \preceq \kappa \bL(G)$.
Next, we will prove that $\Bar{G}_F=(V, \Hat{E}, \Bar{c}_F)$ is a good spectral sparsifier of $\Bar{G}$ using \cref{lemma:FlowRoutingToSpectralApprox}.

First, we show that the 1st condition of \cref{lemma:FlowRoutingToSpectralApprox} holds.
Given any flow $\bf$ on $\Bar{G}$, let $\Bar{\bf} = \Pi_{\Bar{G} \to \Bar{G}_F}(\bf)$.
We want to bound the energy of $\Bar{\bf}$ by the energy of $\bf$.

\paragraph{Total energy across core edges}
For any edge $\Hat{e}$ in the core of $\Bar{G}_F$, it is used only to route edges whose image under $\mathrm{Move}_F$ is $\Hat{e}$, i.e, the set $\mathrm{Move}_F^{-1}(\Hat{e}) \subseteq E(\Bar{G})$.
Thus, we have the followings:
\begin{align*}
    \Bar{f}_{\Hat{e}} &= \sum_{e: \mathrm{Move}_F(e) = \Hat{e}} f_e \\
    \Bar{c}_F(\Hat{e}) &= \sum_{e: \mathrm{Move}_F(e) = \Hat{e}} \Bar{c}(e).
\end{align*}
By Cauchy's Inequality, we have:
\begin{align*}
    \Bar{r}_F(\Hat{e}) \Bar{f}_{\Hat{e}}^2
    &= \left(\sum_{e: \mathrm{Move}_F(e) = \Hat{e}} \Bar{c}(e)\right)^{-1} \left(\sum_{e: \mathrm{Move}_F(e) = \Hat{e}} f_e\right)^2 \\
    &\underbrace{\le}_\text{Cauchy's} \left(\sum_{e: \mathrm{Move}_F(e) = \Hat{e}} \Bar{c}(e)\right)^{-1}\left(\sum_{e: \mathrm{Move}_F(e) = \Hat{e}} \Bar{c}(e)\right) \left(\sum_{e: \mathrm{Move}_F(e) = \Hat{e}} \Bar{c}(e)^{-1}f_e^2\right) \\
    &= \left(\sum_{e: \mathrm{Move}_F(e) = \Hat{e}} \Bar{r}(e)f_e^2\right).
\end{align*}
Therefore, the total energy crossing  $\Bar{G}_F$'s core edges is at most $\bf^\top \bR(\Bar{G}) \bf$.

\paragraph{Total energy across envelope edges}
Given any forest edge $\Bar{e}$, $C(\Bar{e})$ is defined as the component of $F$ containing $\Bar{e}$.
For every edge $e=\Bar{c}(u, v) \in \Bar{G}$, $\Pi_{\Bar{G} \to \Bar{G}_F}(e)$ intersects with the envelope in at most 2 trees $T_1, T_2$ of the forest $F$.
For each maximal tree path $P_i = \Pi(e) \cap T_i, i=1, 2$, we have
\begin{align*}
    \sum_{\Bar{e} \in P_i} \frac{r(\Bar{e})}{\str_F(e; T_i)} = \sum_{\Bar{e} \in P_i} \frac{r(\Bar{e})r(e)}{r(P_i)} = r(e).
\end{align*}
Therefore, we have the following inequality:
\begin{align*}
    \sum_{\Bar{e} \in \Pi(e) \cap F} \frac{r(\Bar{e})}{\str_F(e, C(\Bar{e}))}
    = \sum_i \sum_{\Bar{e} \in P_i} \frac{r(\Bar{e})}{\str_F(e; T_i)}
    \le 2 r(e).
\end{align*}
In the inequality, the resistance $r(\cdot)$ is the on in the original graph $G$ instead of the scaled version $\Bar{G}$.

Recall the flow in $\Bar{G}_F$: $\Bar{\bf} = \Pi_{\Bar{G} \to \Bar{G}_F}(\bf)$.
Due to the construction of $\Bar{G}_F$ and the flow routing, every envelope edge $\Bar{e}$ has the following amount of flow across: $\Bar{f}_{\Bar{e}} = f_{\Bar{e}} + \sum_{e: \Bar{e} \in \Pi(e)}f_e$.
Also, the conductance of every envelope edge $\Bar{e}$ is the same as the one in $\Bar{G}$, thus $\Bar{r}_F(\Bar{e}) = \Bar{r}(\Bar{e}) = r(\Bar{e}) / \kappa.$

Using the inequality and Cauchy's inequality, we can bound the total energy across envelope edges as follows:
\begin{align*}
    \sum_{\Bar{e} \in F} \Bar{r}_F(\Bar{e}) \Bar{f}_{\Bar{e}}^2
    &= \sum_{\Bar{e} \in F} \frac{1}{\kappa} r(\Bar{e}) \left(\sum_{e \neq \Bar{e}: \Bar{e} \in \Pi(e)} f_e + f_{\Bar{e}} \right)^2 \\
    &\underbrace{\le}_\text{Cauchy's} \sum_{\Bar{e} \in F} \frac{1}{\kappa} r(\Bar{e}) \left(\underbrace{\sum_{e \neq \Bar{e}: \Bar{e} \in \Pi(e)} \str_F(e; C(\Bar{e}))}_{\le \kappa} + \kappa \right) \left(\sum_{e \neq \Bar{e}: \Bar{e} \in \Pi(e)} \frac{1}{\str_F(e; C(\Bar{e}))}f_e^2 + \frac{1}{\kappa}f_{\Bar{e}}^2 \right) \\
    &\le \sum_{\Bar{e} \in F}2r(\Bar{e})\left(\sum_{e\neq \Bar{e}: \Bar{e} \in \Pi(e)} \frac{1}{\str_F(e; C(\Bar{e}))}f_e^2 + \frac{1}{\kappa}f_{\Bar{e}}^2\right)\\
    &= \sum_{e \not\in F} \left(\sum_{\Bar{e} \in \Pi(e) \cap F} \frac{2r(\Bar{e})}{\str_F(e; C(\Bar{e}))}\right) f_e^2 + \sum_{\Bar{e} \in F} \frac{2r(\Bar{e})}{\kappa} f_{\Bar{e}}^2 \\
    &\le \sum_{e \not\in F} 4r(e) f_e^2 + \sum_{\Bar{e} \in F} \frac{2r(\Bar{e})}{\kappa} f_{\Bar{e}}^2 \\
    &\le 4 \sum_{e \in \Bar{G}} \Bar{r}(e) f_e^2 = 4 \bf^\top \bR(\Bar{G}) \bf.
\end{align*}

\paragraph{Total energy of $\Bar{\bf}$ within $\Bar{G}_f$}
Combining the energy bound across core and envelope edges, we have $\Bar{\bf}^\top \bR(\Bar{G}_F) \Bar{\bf} \le 5 \bf^\top \bR(\Bar{G}) \bf.$
Therefore, 1st condition of \cref{lemma:FlowRoutingToSpectralApprox} holds.

Similarly, 2nd condition also holds and therefore \cref{lemma:FlowRoutingToSpectralApprox} ensures
\begin{align*}
    0.1 \bL(\Bar{G}_F) \preceq \bL(\Bar{G}) \preceq 10 \bL(\Bar{G}_F).
\end{align*}

Recall that $\bL(G) \preceq \bL(\Bar{G}) \preceq \kappa \bL(G)$ holds.
Plug in the inequality, we have
\begin{align*}
   0.1 \kappa^{-1}\bL(\Bar{G}_F) \preceq \kappa^{-1}\bL(\Bar{G}) \preceq \bL(G) \preceq \bL(\Bar{G}) \preceq 10 \bL(\Bar{G}_F).
\end{align*}
Let $H$ be the graph corresponds to the Laplacian matrix $10 \bL(\Bar{G}_F)$.
Like $\Bar{G}_F$, $H$ is also a $j$-tree.
Also, we have that $\bL(G) \preceq \bL(H)\preceq 100\kappa \cdot \bL(G)$.
\end{proof}

\subsection{Prove \cref{lemma:FindForest} via Tree Decomposition}
\label{sec:FindForest}

In this section, we prove \cref{lemma:FindForest}.
The idea is to first compute a low stretch spanning tree via \cref{lemma:LSST}, and then decompose the spanning tree into a spanning forest.
By carefully decomposing the LSST, we can bound the total local stretch per chunk.
However, the resulting routing does not induce a $j$-tree.
We need to carefully re-route edges without increasing too much stretch.

First, we introduce the notion of \emph{Tree Decomposition}.
\begin{definition}[Tree Decomposition, \cite{SpielmanT03}]
\label{defn:TreeDecomposition}
Given a tree $T$ that spans on a set of vertices $V$, a \emph{$T$-decomposition} of $T$ is a decomposition of $V$ into sets $\WW=\{W_1, \ldots, W_h\}$ such that $V = \bigcup W_i$, $T[W_i]$ is a tree for each $i$ and $|W_i \cap W_j| \le 1$ for all $i \neq j$.

Given an additional set of edges $E$ on $V$, a \emph{$(T, E)$-decomposition} is a pair $(\WW = \{W_1, \ldots, W_h\}, \rho)$ where $\WW$ is a $T$-decomposition and $\rho$ is a mapping from $E$ to one or two sets in $\WW$ so that for every $e=uv \in E$,
\begin{enumerate}
    \item $\rho(uv) = \{W_i\}$, then $u, v \in W_i$, and
    \item $\rho(uv) = \{W_i, W_j\}$, then $u \in W_i$ and $v \in W_j$.
\end{enumerate}
\end{definition}


\begin{definition}[Refined Decomposition]
\label{defn:TreeRefinedDecomposition}
Given a tree $T$ that spans on a set of vertices $V$, a $T$-decomposition $\WW=\{W_1, \ldots, W_h\}$ is also a \emph{$T$-refined-decomposition} if for each $i$, $W_i$ contains at most 2 vertices that are also in other sets in the decomposition, i.e,
\[
    |W_i \cap \bigcup_{j \neq i}W_j| \le 2.
\]

Similarly, given an additional set of edges $E$ on $V$, a \emph{$(T, E)$-refined-decomposition} is a pair $(\WW = \{W_1, \ldots, W_h\}, \rho)$ which is a $(T, E)$-decomposition and $\WW$ being a $T$-refined-decomposition.
\end{definition}

$T$-decomposition can be computed in linear time as shown in \cite{SpielmanT03}.
It is done by performing depth-first search on the tree $T$.
Given a $T$-decomposition $\WW$, one can extend $\WW$ to a $T$-refined-decomposition in linear time.
Vertices that shows up more than once in $\WW$ are marked as terminals.
Branching vertices are also add to the set of terminals.
By splitting the tree $T$ by these terminal vertices, one have the $T$-refined-decomposition.
Note that the procedure produce a decomposition containing at most $2|\WW|$ sets.

\begin{lemma}
\label{lemma:TreeRefinedDecomposition}
There exists a linear time algorithm, which we invoke with the syntax
\[
    (\WW, \rho) \coloneqq \textsc{Decompose}(G, T, w, j)
\]
that on input of a graph $G=(V,E,c)$, a spanning tree $T$ of $G$, a function $w: E \to \Real^+$, and an integer $j \ge 1$, outputs a $(T, E)$-refined-decomposition $(\WW = \{W_1, \ldots, W_h\}, \rho)$ such that
\begin{enumerate}
    \item $|\WW| \le j$, and
    \item for all $W_i$ with $|W_i| > 1$,
    \[
        \sum_{e: W_i \in \rho(e)} w(e)\le \frac{20}{j} w(E)
    \]
\end{enumerate}
\end{lemma}


\begin{algorithm}
\caption{Algorithm for \cref{lemma:FindForest}}
\label{algo:FindForest}
\begin{algorithmic}[1]
    \Procedure{FindForest}{$G=(V,E,c), j$}
        \State $T \coloneqq \mathrm{LSST}(G)$.
        \State $(\WW, \rho) = \textsc{Decompose}(G, T, \str_T, j + 1)$
        \State Let $C$ be the set of boundary vertices of the $T$-refined-decomposition $\WW$.
        \State $T^- \coloneqq \phi$
        \For{$\Bar{e} \in T$}
            \State Compute $\Cong_T(\Bar{e}) = \sum_{e: \Bar{e} \in T(e)} c(e)$
        \EndFor
        \For{$W_i \in \WW$ such that $|W_i \cap C| = 2$}
            \State $\{u, v\} = W_i \cap C$.
            \State $e \coloneqq \arg\min_{\Bar{e} \in T[u, v]} \Cong_T(\Bar{e})$.
            \State Add $e$ to $T^-$.
        \EndFor
        \State $T \setminus T^-$ is a rooted spanning forest of $G$ with $C$ being the set of roots.
        \State \Return $(F, C)$.
    \EndProcedure
\end{algorithmic}
\end{algorithm}

\textbf{Correctness of \cref{algo:FindForest}:}
To analyze the algorithm, we can view it as a 2 phases procedure.
Phase 1 (before line 5) routes each edge $e$ along $T(e)$, the tree path connecting endpoints of $e$.
Phase 2 removes edges from $T$ and re-route edges whose routing got affected.

\cref{lemma:TreeRefinedDecomposition} guarantees the size of $\WW$ is at most $j + 1$.
Thus, the size of $C$ which is the number of boundary vertices is at most $j$.
Therefore, 1st condition of \cref{lemma:FindForest} is satisfied.

Next, we will bound the total local stretch of every component of the forest $F$.
Given any component of $F$, say $S$, it lies entirely in only one tree piece of $\WW$, say $W_i$.
The following claim classified the set of edges which contribute to the total local stretch on $S$.
\begin{claim}
\label{claim:localEdges}
For any $e=c(u, v) \in G$ with $\str_F(e; S) > 0$, $W_i$ must be an element of $\rho(e)$ and $T(e)$ contains edges in $W_i$.
\end{claim}
\begin{proof}
Let $P = P_F(e) \cap S$, the routing path of $e$ inside $S$.
Assume one of $P$'s endpoint being $u$ without loss of generality.

If $P$'s other endpoint is $v$, $P$ is exactly the tree path $T(e)$.
None of $P$'s interior vertex is a root in $F$, and every edge of $P$ lies in $W_i$.
Since $P$ contains at least one edge, $P$ cannot have both endpoints being roots of $F$.
Therefore, either $u$ or $v$ is interior in $W_i$ and thus $W_i \in \rho(e)$.

If $P$'s other endpoint is not $v$, it must be a root of $F$.
By the $P$'s non-emptyness, $u$ must be an interior vertex in $W_i$ and thus $W_i \in \rho(e)$.
Also, every edge of $P$ lies in $W_i$.

To prove that $T(e)$ contains edge in $W_i$, consider only the case where $P$'s other endpoint not being $v$ and $P \cap T(e) = \phi$.
Previous argument ensures that $u$ is an interior vertex of $W_i$.
Thus, every incident edge of $u$ lies in $W_i$ and $T(e) \cap W_i$ cannot be empty.
\end{proof}

Let $D$ be the set of edges satisfying the necessary condition in \cref{claim:localEdges}.
That is, $D = \{e \in E \mid W_i \in \rho(e), T(e) \cap W_i \neq \phi\}$.

Next, we do case analysis for bounding the total local stretch of $S$.

\textbf{Case I: $W_i$ contains only 1 boundary vertex $s$.}

In this case, we will prove that $\str_F(e; S) \le \str_T(e)$ for any $e \in D$.
Since $W_i$ has only 1 boundary vertices, no edge of $W_i$ is removed.
Thus, for any edge $e=c(u, v) \in D$ with $u \in W_i$, $\rho_F(u)$ is $s$, the only boundary vertex of $W_i$.

If $\rho_F(v)$ is also $s$, $\str_F(e; S) = r(T(e)) / r(e)$.
Otherwise, $T(e)$ starts from $u$, leaves $W_i$ at the only boundary vertex $s$, and then goes to $v$.
Thus, $T[u, s]$ is a subpath of $T(e)$ and $\str_F(e; S) = r(T[u, s]) / r(e) \le r(T(e)) / r(e) = \str_T(e)$.
This fact yields
\begin{align*}
    \sum_{e \in D} \str_F(e; S)
    \le \sum_{e \in D} \str_T(e)
    \le \sum_{e: W_i \in \rho(e)} \str_T(e)
    \le \frac{20}{j}\str_T(E) = O\left(\frac{m \log n \log\log n}{j}\right).
\end{align*}

\textbf{Case II: $W_i$ contains 2 boundary vertices $s$ and $t$.}

In this case, we will prove that:
\begin{align*}
    \sum_{e \in D} \str_F(e; S) = O\left(\sum_{e: W_i \in \rho(e)} \str_T(e)\right).
\end{align*}
Let $e_1, e_2, \ldots, e_k$ and $(s=u_0), u_1, \ldots, (u_k=t)$ be the edges and vertices of the path $T[s, t]$.
Let $e_{i^*}$ be the edge removed in line 10.
For every $i > 0$, define $D_i = \{e \in D \mid e_i \in T(e)\}$.
By the definition of $\Cong_T(\cdot)$, we have that $\Cong_T(e_i) = c(D_i)$, sum of conductance of edges in $D_i$.
We also define $D_0 = D \setminus \bigcup_{i=1} D_i$, the set of edges whose corresponding tree paths do not intersect with the $st$-tree-path.

\begin{claim}
\label{claim:edgeRemovalBound}
Given any edge $e=c(u, v) \in D$ with $u \in W_i$, if we remove $e_j$ in line 10, the following holds:
\begin{align*}
    \str_F(e; S) &\le \str_T(e), &\text{if } e_j \in T(e) \\
    \str_F(e; S) &\le \str_T(e) + \frac{r(T[s, t] \setminus T(e))}{r(e)}, &\text{otherwise}.
\end{align*}
\end{claim}
\begin{proof}
We prove the claim via case analysis:
\begin{itemize}
    \item
    If $e_j$ is not in $T(e)$, the part of $P_F(e)$ containing $u$ is a subpath of $T(e)$.
    Therefore, $\str_F(e; S) \le r(T(e)) / r(e)$ holds.
    
    \item
    Otherwise, $W_i$ is split into 2 disconnected pieces $S$ and $T$.
    Observe that $S$ and $T$ partitions $W_i$'s vertex set.
    Let $s$ (and $t$) be the boundary vertex in $S$ (and $T$, respectively).
    
    We assume without loss of generality that $u \in S$.
    $T(e)$ starts with a tree-path $T[u, u_{j-1}]$, the endpoint of $e_j$ in $S$.
    Since $e_j \in T(e)$, the other endpoint of $T(e)$ must be in other component of $F$.
    Therefore, $P_F(e) \cap S$ is the unique $us$-tree path, $T[u, s]$.
    
    Let $u_k$ be the first vertex on the path $T(e) \cap T[s, t]$.
    $T(e)$ can be expressed as $T[u, u_k]$ followed by $T[u_k, u_{j-1}]$.
    $P_F(e) \cap S$ can be expressed as $T[u, u_k]$ followed by $T[u_k, s]$.
    Since $T[u_k, s] \subseteq T[s, t] \setminus T(e)$, $r(P_F(e) \cap S)$ can be bounded by:
    \begin{align*}
        r(P_F(e) \cap S)
        = r(T[u, u_k]) + r(T[u_k, s])
        \le r(T(e)) + r(T[s, t] \setminus T(e)).
    \end{align*}
    
    
\end{itemize}
\end{proof}

The removal of edge $e_{i^*}$ from $T[s, t]$ in line 10 of \cref{algo:FindForest} and the claim above bound the total local stretch as follows:
\begin{align}
\label{eq:totalLocalStretch}
    \sum_{e \in D} \str_F(e; S) &\le \sum_{e \in D} \frac{r(T(e))}{r(e)} + \sum_{e \in D_{i^*}} r(T[s, t]) c(e) - \sum_{e \in D_{i^*}} r(T(e) \cap T[s, t]) c(e).
\end{align}
Next, we bound each term in (\ref{eq:totalLocalStretch}).
\begin{enumerate}
    \item $\sum_{e \in D} \frac{r(T(e))}{r(e)}$: It is at most $\kappa$ from the guarantee of \cref{lemma:TreeRefinedDecomposition}.
    
    \item $\sum_{e \in D_{i^*}} r(T[s, t]) c(e)$:
    The bound on total stretch from \cref{lemma:TreeRefinedDecomposition} and the optimality from the choice of $i^*$ yield
    \begin{align*}
        \kappa
        &\ge \sum_{e \in D} r(T(e) \cap T[s, t])c(e) \\
        &= \sum_{e \in D} \sum_{e_i \in T(e) \cap T[s, t]}r(e_i) c(e) \\
        &= \sum_{e_i \in T[s, t]} \sum_{e \in D_i} r(e_i) c(e) \\
        &= \sum_{e_i \in T[s, t]} r(e_i) c(D_i)
        \ge \sum_{e_i \in T[s, t]} r(e_i) c(D_{i^*}) = r(T[s, t]) c(D_{i^*}),
    \end{align*}
    where the second equality comes from the definition of $D_i$.

    \item $\sum_{e \in D_{i^*}} r(T(e) \cap T[s, t]) c(e)$: Since $D_{i^*} \subseteq D$, this term can be bounded by $\sum_{e \in D} r(T(e))c(e)$ which is at most $\kappa$.
\end{enumerate}
Therefore, one can bound the LHS of (\ref{eq:totalLocalStretch}) by $3\kappa$.

Combining Case I and II, we have that for any component $S$ of $F$:
\begin{align*}
    \sum_{e \in G} \str_F(e; S) \le 3 \kappa = O\left(\frac{m \log n \log\log n}{j}\right).
\end{align*}
The correctness of \cref{algo:FindForest} is proved.

\textbf{Time Complexity of \cref{algo:FindForest}}:
The bottleneck of the procedure is computing the LSST for $G$.
Such routine takes $O(m \log n \log\log n)$-time via \cref{lemma:LSST}.
Thus, the whole algorithm runs in $O(m \log n \log\log n)$-time.

\begin{proof}[Proof of \cref{lemma:FindForest}]
The lemma is proved via above discussions.
\end{proof}

\section{Diffusion on J-Trees}
\label{sec:JTreeSolve}

In this section, we present tools and combine them to prove \cref{lemma:JTreeSolve}.

The first tool eliminates vertices of degree 1.
Vertex elimination routine is critical in designing recursive algorithms, e.g. \cite{SpielmanT03, Peng16}.
Via reducing the number of vertex, one can reduce the problem to a smaller graph.
The process eliminates degree 1 vertices one at a time.
Meanwhile, it updates VWF of its only neighbor.

When the input is a $j$-tree, entire \emph{envelope} will be eliminated.
So one can reduce the problem on a graph with $n$ vertices to the one with $j=o(n)$ vertices.

The tool is formalized as the following lemma and will be proved in \cref{sec:vtxElimination}:
\begin{lemma}
\label{lemma:vtxElimination}
There exists an algorithm (\cref{algo:vtxElimination}) invoked with the syntax
\[
    (\HH, \MM) = \textsc{VertexElimination}(\GG)
\]
that on input of a diffusion instance $\GG$ with $m$ edges and VWFs of total size $S=\poly(n)$, outputs another diffusion instance $\HH$ such that
\begin{enumerate}
    \item $H$, the underlying graph of $\HH$, is a subgraph of $G$.
    \item $\HH \preceq_1 \GG$ with the mapping $\MM$.
    $\MM$ can be applied in $O(n \log n)$ time and is implicitly maintained using $O(n \log n)$-space.
    
    \item $\GG \preceq_1 \HH$ with the mapping $\Bar{\MM}(\bx)=\bx[V(H)]$.
    \item Total size of $\HH$'s VWFs is $O(S)$.
    \item The algorithm runs in $O((S+n) \log^2 n)$-time.
\end{enumerate}
Additionally, if $G$, the underlying graph of $\GG$, is a $j$-tree with $m_j$ core-edges, the resulting $\HH$ has at most $j$ vertices and $m_j$ edges.
\end{lemma}

\cref{fact:embedApprox} and the fact that $\GG \preceq_1 \HH \preceq_1 \GG$ yield the following claim:
\begin{corollary}
\label{coro:vtxElimApprox}
Let $(\HH, \MM) = \textsc{VertexElimination}(\GG)$.
For any $\alpha$-approximated optimal potential $\bx$ of $\HH$, $\MM(\bx)$ is an $\alpha$-approximated optimal for $\GG$.
\end{corollary}

Therefore, diffusion problems on a $j$-tree can be reduced to ones defined on its core graph.
\begin{corollary}
\label{coro:JTreeSolve}
Given any $j$-tree $G$ with $m_j$ core edges.
One can $(1+\eps)$-approximately solve any diffusion instance on $G$ with total VWF size $S$ in time:
\begin{align*}
    T\left(|E(G)|, S, 1 + \eps\right) = O\left(S \log^2 n\right) + T\left(m_j, O(S), 1 + \eps\right).
\end{align*}
\end{corollary}

The second tool approximates any VWF with a $O(\log n)$-sized one.
It is formalized as the following lemma and will be proved in \cref{sec:CompressVWF}.
\begin{lemma}
\label{lemma:CompressVWF}
There exists an algorithm (\cref{algo:CompressVWF}) invoked with syntax
\[
    \ftil \coloneqq \textsc{Compress}(f)
\]
that on input of a VWF $f$ of size $S$, outputs another VWF $\ftil$ of size $O(\log n)$ in $O(S)$-time such that (1.) $\dom(f)=\dom(\ftil)$, and (2.) $2 f(x/2) \le \ftil(x) \le f(x)$ holds for any $x \in \dom(f)$.
\end{lemma}

Given any diffusion instance $\GG$, applying \cref{lemma:CompressVWF} to every VWF of $\GG$ yields the following corollary:
\begin{corollary}
\label{coro:CompressVWF}
There exists an algorithm invoked with syntax
\[
    \widetilde{\GG} \coloneqq \textsc{Compress}(\GG)
\]
that on input of a diffusion instance $\GG$ with $m$ edges and VWFs of total size $S$, outputs another diffusion instance $\widetilde{\GG}$ such that
\begin{enumerate}
    \item $\widetilde{\GG} \preceq_2 \GG \preceq_1 \widetilde{\GG}$ holds with identity mapping.
    \item Each VWF of $\widetilde{\GG}$ has size $O(\log n)$.
    \item The algorithm runs in $O(S)$-time.
\end{enumerate}
\end{corollary}


\begin{algorithm}
\caption{Algorithm for \cref{lemma:JTreeSolve}}
\label{algo:JTreeSolve}
\begin{algorithmic}[1]
    \Procedure{JTreeSolve}{$\HH, \eps, \OO$}
        \State $(\CC, \MM) \coloneqq \textsc{VertexElimination}(\HH)$
        \State $\bx = \textsc{Iter}(\CC, \eps, (\OO \circ \textsc{Compress}), 4)$
        \State \Return $\MM(\bx)$
    \EndProcedure
\end{algorithmic}
\end{algorithm}

The proof for \cref{lemma:JTreeSolve} combines \emph{Vertex Elimination} (\cref{coro:JTreeSolve}), \emph{VWF Compression}, and \emph{Iterative Refinement} (\cref{lemma:IterRefine}).
\begin{proof}[Proof of \cref{lemma:JTreeSolve}]
The oracle $\OO$ computes 2-approximate optimal potential given any instance.
By \cref{fact:embedApprox} and \cref{coro:CompressVWF}, $\by = \OO(\textsc{Compress}(\CC))$ is a 4-approximate optimal potential of $\CC$.
\cref{lemma:IterRefine} guarantees that $\bx$ in line 3 is a $(1+\eps)$-approximate optimal potential of $\CC$.

Since $H$ is a $j$-tree with $O(j \log n)$ edges in its core, computing $\bx$ takes
\begin{align*}
    O\left(\log\left(\frac{1}{\eps}\right)\left(m + S + T\left(O(j \log n), O(j \log n), 2\right)\right)\right),
\end{align*}
-time.
Plus the time for eliminating vertices and computing $\MM(\bx)$, \cref{algo:JTreeSolve} has time compelxity:
\begin{align*}
    O\left(m\log^2 n + \log\left(\frac{1}{\eps}\right)\left(m + T\left(O(j \log n), O(j \log n), 2\right)\right)\right)
\end{align*}
\end{proof}

\subsection{Vertex Elimination}
\label{sec:vtxElimination}

In this section, we prove Lemma~\ref{lemma:vtxElimination} via \cref{algo:vtxElimination}.
It eliminates degree 1 vertices one at a time.
By eliminating a degree 1 vertex $u$, we update the VWF corresponding to $u$'s only neighbor $v$ to encode the optimal contribution of $u$.
Suppose $u$ is connected to $v$ with edge $e=c(u, v)$, we replace $f_v$, the VWF of $v$, with $f^{\mathtt{new}}_v$, defined as
\begin{align*}
    f^{\mathtt{new}}_v(x) = f_v(x) + \min_{y \ge b_u} \frac{1}{2} c(x - y)^2 + f_u(y),~x \ge b_v.
\end{align*}
We will see that $f^{\mathtt{new}}_v$ is also a VWF.
In addition, we maintain a mapping $\Pi$ to recover the value on $u$ using the value on $v$.
That is,
\begin{align*}
    \Pi(x) = \argmin_{y \ge b_u} \frac{1}{2} c(x - y)^2 + f_u(y).
\end{align*}

For efficiency, we use sequence data structures such as binary search trees to maintain VWFs during the elimination.
The data structure maintaining VWFs is formalized as the following lemma and will be proved in \cref{sec:VWFDS}.

\begin{lemma}
\label{lemma:VWFDS}
There is a data structure $\DD$ maintaining VWFs subject to the following operations:
\begin{enumerate}
    \item $\mathrm{Initialize}(f)$:
    Initialize $\DD$ given a VWF $f$ in time $O(|f|)$.
    
    \item $\mathrm{Add}(\DD_g)$:
    Given the same type of data structure $\DD_g$ maintaining another VWF $g$, update the underlying VWF to be $f + g$ in time $O(\min\{|f|, |g|\}\log n)$ .
    The resulting VWF has size at most $|f| + |g|$.
    
    \item $\mathrm{Lift}(c)$:
    Update the underlying VWF with another VWF:
    \begin{align*}
        \Hat{f}(y) = \min_{x \in \dom(f)} \frac{1}{2}c(y-x)^2 + f(x).
    \end{align*}
    The operation can be done in $O(\log |f|)$-time and the resulting VWF has size at most $|f| + 1$.
    
    \item $\mathrm{OptimalX}(y)$:
    If the latest update to $\DD$ is $\mathrm{Lift}(c)$, output the corresponding optimal $x$ given $y$ in $O(\log |f|)$-time.
    
    \item $\mathrm{VWF}()$:
    Output the canonical representation of VWF, the list $\{(s_i, r_i, a_i, b_i)\}$, in linear time.
\end{enumerate}
Additionally, $\DD$ can be implemented in pointer machine model and can be persistently maintained with constant overhead (\cite{DSST89}).
\end{lemma}

\begin{algorithm}
\caption{Algorithm for \cref{lemma:vtxElimination}}
\label{algo:vtxElimination}
\begin{algorithmic}[1]
    \Procedure{VertexElimination}{$\GG=(G, \bf, \bb)$}
        \State $t \coloneqq 0$
        \State Initialize $\MM^{(0)}$ as the identity function on $\Real^{V(G)}$
        \State $G^{(0)} = G$
        \For{$u \in V(G)$}
            \State $D_u = \DD.\mathrm{Initialize}(f_u)$
        \EndFor
        \While{there is a vertex $u$ of degree 1 in $G^{(t)}$}
            \State Let $v$ be $u$'s only neighbor connected by edge $e=c(u, v)$
            \State $G^{(t+1)} \coloneqq G^{(t)} \setminus \{u\}$
            \State $D_u.\mathrm{Lift}(c)$
            \State $D_v.\mathrm{Add}(D_u)$
            \Comment{Update $f_v(x) = f_v(x) + \min_{y \ge b_u}\frac{1}{2}c(x-y)^2 + f_u(y)$}
            \State Make a snapshot of $D_v$
            \State Define $\Pi^{(t+1)}:\Real^{V^{(t)}} \to \Real^{V^{(t+1)}}$ as: 
            \begin{align*}
                \Pi^{(t+1)}(\bx)_i \coloneqq \begin{cases}
                    x_i, &i \neq u \\
                    D_v.\mathrm{OptimalX}(x_v), &i = u
                \end{cases}
            \end{align*}
            \State $\MM^{(t+1)} \coloneqq \MM^{(t)} \circ \Pi^{(t+1)}$
            \State $t \coloneqq t + 1$
        \EndWhile
        \State $\GG^{(t)} \coloneqq (G^{(t)}, \{D_u.\mathrm{VWF}() \mid u \in V^{(t)}\}, \bb^\GG[V^{(t)}])$
        \State \Return $(\GG^{(t)}, \MM^{(t)})$.
    \EndProcedure
\end{algorithmic}
\end{algorithm}

\begin{proof}[Proof of \cref{lemma:vtxElimination}]
See \cref{algo:vtxElimination} for the detailed pseudocode proving the lemma.

\textbf{Correctness:}
The correctness can be proved via induction.
At any time $t$, we define the instance $\GG^{(t)}$ the way as line 17.
We claim the followings:
\begin{claim}
At any time $t$, $\GG^{(t+1)} \preceq_1 \GG^{(t)}$ holds with mapping $\Pi^{(t+1)}$.
Also, $\GG^{(t)} \preceq_1 \GG^{(t+1)}$ holds with mapping $\bx \mapsto \bx[V^{(t+1)}]$.
\end{claim}
\begin{proof}
Let $u$ be the vertex removed from $G^{(t)}$, and $v$ be $u$'s only neighbor connected by edge $e=c(u, v)$.
Let $\bx$ be any feasible potential for $\GG^{(t+1)}$ and $\Bar{\bx} = \Pi^{(t+1)}(\bx)$.
We have:
\begin{align*}
    \EE^{\GG^{(t+1)}}(\bx) 
    &= \frac{1}{2}\bx^\top \bL(G^{(t+1)}) \bx + \sum_{w \neq v \in G^{(t+1)}} f^{(t+1)}_w(x_w) + f^{(t+1)}_v(x_v) \\
    &= \frac{1}{2}\bx^\top \bL(G^{(t+1)}) \bx + \sum_{w \neq v \in G^{(t+1)}} f^{(t)}_w(x_w) + f^{(t)}_v(x_v) + \min_{y \ge b_u} \frac{1}{2}c(x_v - y)^2 + f^{(t)}_u(y) \\
    &= \frac{1}{2}\bx^\top \bL(G^{(t+1)}) \bx + \sum_{w \neq v \in G^{(t+1)}} f^{(t)}_w(x_w) + f^{(t)}_v(x_v) + \frac{1}{2}c(x_v - \Bar{x}_u)^2 + f^{(t)}_u(\Bar{x}_u) \\
    &= \frac{1}{2}\Bar{\bx}^\top \bL(G^{(t)}) \Bar{\bx} + \sum_{w \in G^{(t)}} f^{(t)}_w(\Bar{x}_w) = \EE^{\GG^{(t)}}(\Bar{\bx}). 
\end{align*}
Thus, $\GG^{(t+1)} \preceq_1 \GG^{(t)}$ holds.

For the other part, let $\by$ be any feasible potential for $\GG^{(t)}$.
Let $\Bar{\by} = \by[V^{(t)}]$, $\by$ with $u$'s coordinate removed.
We have:
\begin{align*}
    \EE^{\GG^{(t)}}(\by)
    &= \frac{1}{2}\by^\top \bL(G^{(t)}) \by + \sum_{w \in G^{(t)}} f^{(t)}_w(y_w) \\
    &= \frac{1}{2}\Bar{\by}^\top \bL(G^{(t+1)}) \Bar{\by} + \frac{1}{2}c(\by_v - \by_u)^2 + f^{(t)}_v(y_v) + f^{(t)}_u(y_u) + \sum_{w \neq u, v \in G^{(t)}} f^{(t)}_w(y_w) \\
    &\ge \frac{1}{2}\Bar{\by}^\top \bL(G^{(t+1)}) \Bar{\by} + f^{(t+1)}_v(y_v) + \sum_{w \neq u, v \in G^{(t)}} f^{(t)}_w(y_w) = \EE^{\GG^{(t+1)}}(\Bar{\by}).
\end{align*}
Thus, $\GG^{(t)} \preceq_1 \GG^{(t+1)}$ holds.
\end{proof}

Using the claim and \cref{fact:embedTrans}, we show inductively that $\GG^{(t)} \preceq_1 \GG$ holds with mapping $\MM^{(t)}$ and $\GG \preceq_1 \GG^{(t)}$ with mapping $\bx \mapsto \bx[V^{(t)}]$.

\textbf{Running Time:}
Let $m$ and $n$ be the number of edges and vertices in $G$.
Let $S=\poly(n)$ be the total size of VWFs of $\GG$.
Initializing data structures for each VWF (line 5, 6) takes $O(n + S)$-time by \cref{lemma:VWFDS}.

Next, we analyze the total time spent in the main loop.
Note that the time complexity for each operation of $D_u$ is related to the size of its underlying VWF.
$D_u.\mathrm{Lift}(c)$ increases the size by at most 1 and takes $O(\log n)$-time.
$D_u.\mathrm{Add}(D_v)$ merges 2 VWF and takes time proportional to the size of the smaller VWF.
The loop contains at most $n$ iterations since each iteration remove 1 vertex.
A classical amortized analysis shows that the total time is $O((S + n)\log^2 n)$.

\textbf{Time Complexity of Applying $\MM^{(t)}$:}
$\MM^{(t)}$ is represented by a series of mappings $\Pi^{(1)}, \Pi^{(2)}, \ldots, \Pi^{(t)}$.
For each mapping $\Pi^{(i)}$, we store only pointer to the snapshot of $D_v$ made in line 12.
Using persistance technique, only $O((S + n)\log n)$ space is needed.
Computing $\Pi^{(i)}(\bx)$ only takes $O(\log n)$-time by a single call of $D_v.\mathrm{OptimalX}(x_v)$.

To compute $\MM^{(t)}(\bx)$ given any $\bx$, we compute $\Pi^{(1)} \circ \Pi^{(2)} \circ \ldots \circ \Pi^{(t)} (\bx)$ in $O(t \log n)$-time.
Since $t \le m$, it is bounded by $O(m \log n)$-time.
\end{proof}

\section{Vertex Weighting Functions}
\label{sec:VWFtools}

In this section, we prove \cref{lemma:CompressVWF} and \cref{lemma:VWFDS}.

\subsection{Compressing VWFs}
\label{sec:CompressVWF}


\subsubsection*{Preliminaries}

Recall the naive representation of a $k$-sized VWF $f$ being a list $\{(s_i, r_i, a_i, b_i)\}_{i=0}^{k-1}$ where $f(x) = (1/2)r_ix^2 + a_ix + b_i, x \in [s_i, s_{i+1})$.
Conventionally, $s_k$ is defined as $\infty$ and $\dom(f) = [s_0, \infty)$.

The \emph{Bregman Divergence} of a function $f$ around $y$ is defined as $B_f(x; y) = f(x) - f(y) - \grad f(y)^\top (x-y)$.
We write $B_f(x)$ to denote $B_f(x; 0)$.
If $f$ is convex, $B_f$ is non-negative for any $x$.
Using the closure property of VWF, we have the following fact:
\begin{fact}
If $f$ is a VWF of size $k$, $B_f(x; y)$ is also a VWF of size at most $k$ for any $y \in \dom(f)$.
\end{fact}

For any number $x > 0$ and $c > 1$, $[x]_c$ is the smallest integral power of $c$ not smaller than $x$.
For $x = 0$, $[0]_c$ is defined as $0$.
For $x < 0$, $[x]_c$ is defined as $-[-x]_c$.
\begin{fact}
\label{fact:roundSize}
Fix a constant $c > 1$.
For any set of numbers $S = \{x\}$ with $\max_S |x| / \min_S |x| = \poly(n)$, we have $|S_c| = O(\log n)$ where $S_c = \{[x]_c \mid x \in S\}$.
\end{fact}

\subsubsection*{Compression Scheme}

The idea of compressing a VWF is (1) decomposing $f$ of size $k$ into $k$ VWFs of size at most 2 and (2) rounding the split point of every VWF.
After rounding split points, these $k$ VWFs have only $O(\log n)$ distinct split points in total.
Their summation gives a good approximation to the original VWF.

The following observation motivates decomposing $B_f$ instead of the original VWF.
\begin{observation}
\label{obs:BDDecomp}
Let $f=\{(s_i, r_i, a_i, b_i)\}_{i=0}^{k-1}$ be any VWF.
For every $i = 1, 2, \ldots, k-1$, define $d_i = r_i - r_{i+1}$.
We have
\begin{enumerate}
    \item $d_i \ge 0, \forall i$.
    \item $B_f(x) = \int_0^x \int_0^u \sum_{i=1}^{k-1} d_i \chi_{(s_0, s_{i+1})}(v) dv du.$
\end{enumerate}
\end{observation}
\begin{proof}
The definition of VWF says $f'$ is continuous piece-wise linear, concave, and increasing.
Therefore, $\{r_i\}$ must be a non-increasing sequence, and
\begin{align*}
    f''(x) = r_i, x \in (s_i, s_{i+1}).
\end{align*}

Since $f'$ is continuous and $f''$ is defined except finite number of points, we have
\begin{align*}
    B_f(x) = f(x) - f'(0)x - f(0) = \int_0^x f'(u) - f'(0) du = \int_0^x\int_0^u f''(v) dv du.
\end{align*}

Given any $x \in (s_i, s_{i+1})$, we can telescope the following:
\begin{align*}
    \sum_{j=1}^{k-1} d_j \chi_{[s_0, s_{j+1})}(x) = \sum_{j=i}^{k-1} d_j = \sum_{j=i}^{k-1} r_j - r_{j+1} = r_i - r_k = r_i.
\end{align*}
Therefore, $\sum_{i=1}^{k-1} d_i \chi_{(s_0, s_{i+1})}(x)$ differs from $f''$ at finite number of points.
We can further express $B_f$ by
\begin{align*}
    B_f(x) = \int_0^x\int_0^u f''(v) dv du = \int_0^x \int_0^u \sum_{i=1}^{k-1} d_i \chi_{(s_0, s_{i+1})}(v) dv du.
\end{align*}
\end{proof}

\begin{observation}
\label{obs:Size1VWF}
Given any $l, r$ with $l \le 0$, define $B_{l, r}(x) = \int_0^x \int_0^u\chi_{(l, r)}(v) dv du$.
If $r \ge 0$, we have
\begin{align*}
    B_{l, r}(x) &= \begin{cases}
        \frac{1}{2}x^2 &, x \in [l, r) \\
        rx - \frac{r^2}{2} &, x \in [r, \infty).
    \end{cases}
\end{align*}
Otherwise ($r < 0$), we have
\begin{align*}
    B_{l, r}(x) &= \begin{cases}
        \frac{1}{2}(x-r)^2 &, x \in [l, r) \\
        0 &, x \in [r, \infty).
    \end{cases}
\end{align*}
Thus, $B_{l, r}$ is a VWF of size 2.
\end{observation}

\cref{obs:BDDecomp} and \cref{obs:Size1VWF} suggest that $B_f = \sum_{i=1}^{k-1}d_i B_{s_0, s_{i+1}}$ is a possible decomposition into small-sized VWFs.

\begin{definition}
\label{defn:BDDecomp}
Given any VWF $f$ of size $k$, \emph{Decomposition of $B_f$} is the set of tuples $\{(d_i, B_{s_0, s_{i+1}})\}_{i=1}^{k-1}$.
\end{definition}

Discussion above implies a linear time algorithm for computing such decomposition.
We formalize the result as follows:
\begin{lemma}
\label{lemma:DecomposeVWF}
There exists an algorithm invoked with syntax
\[
    \{(d_i, B_{s_0, s_{i+1}})\}_{i=1}^{k-1} \coloneqq \textsc{Decompose}(f)
\]
that on input of a VWF $f$ of size $k$, outputs a \emph{Decomposition of $B_f$} in $O(k)$-time.
\end{lemma}

The following lemma formalizes the rounding scheme for each $B_{l, r}$.
It will be proved in \cref{sec:Round1VWF}.
\begin{lemma}
\label{lemma:Round1VWF}
Given any $B_{l, r}$ and $d > 0$, let $\Bar{r} = [r]_{1.1}$ and $\Bar{d} = dr/\Bar{r}$, we have
\begin{align*}
    2dB_{l, r}(x/2) \le \Bar{d}B_{l, \Bar{r}}(x) \le dB_{l, r}(x), \forall x \ge l.
\end{align*}
\end{lemma}

\begin{algorithm}
\caption{Algorithm for \cref{lemma:CompressVWF}}
\label{algo:CompressVWF}
\begin{algorithmic}[1]
    \Procedure{CompressVWF}{$f = \{(s_i, r_i, a_i, b_i)\}_{i=0}^{k-1}$}
        \State $\ftil(x) = f(0) + f'(0) x$
        \State $\{(d_i, B_{s_0, s_{i+1}})\}_{i=1}^{k-1} \coloneqq \textsc{Decompose}(f)$
        \For{$i=1, 2, \ldots, k-1$}
            \State $\Bar{s} = [s_{i+1}]_{1.1}$
            \State $\Bar{d} = d_i s_{i+1} / \Bar{s}$
            \State $\ftil = \ftil + \Bar{d} B_{s_0, \Bar{s}}$
        \EndFor
        \State \Return $\ftil$
    \EndProcedure
\end{algorithmic}
\end{algorithm}

\begin{proof}[Proof of \cref{lemma:CompressVWF}]
First, we bound the size of $\ftil$.
Note that
\begin{align*}
    \ftil(x) = f(0) + f'(0)x + \sum_{i=1}^{k-1} \frac{d_i s_{i+1}}{[s_{i+1}]_{1.1}} B_{s_0, [s_{i+1}]_{1.1}}(x).
\end{align*}
The number of split points of $\ftil$ is therefore $O(\log n)$ by \cref{fact:roundSize}.

Next, we bound the time for constructing $\ftil$.
Computing $\textsc{Decompose}(f)$ takes linear time.
Since $\Bar{s} = [s_{i+1}]_{1.1}$ is non-decreasing in $i$, adding $\Bar{d} B_{s_0, \Bar{s}}$ to $\ftil$ can be done by only examine the last element of the list representing $\ftil$.
Thus, computing $\ftil$ takes only linear time.

Last, we argue the approximation quality of $\ftil$.
The domain of $\ftil$ is the same as $f$.
For any $x \in \dom(f)$, we have
\begin{align*}
    \sum_{i=1}^{k-1} d_i B_{s_0, s_{i+1}}(x)
    \ge \sum_{i=1}^{k-1} \frac{d_i s_{i+1}}{[s_{i+1}]_{1.1}} B_{s_0, [s_{i+1}]_{1.1}}(x) 
    \ge \sum_{i=1}^{k-1} 2d_i B_{s_0, s_{i+1}}(x/2)
\end{align*}
by \cref{lemma:Round1VWF}.
Since $f(x) = f(0) + f'(0)x + \sum_{i=1}^{k-1} d_i B_{s_0, s_{i+1}}(x)$, we have
\begin{align*}
    f(x) \ge \ftil(x) \ge 2f(x/2), \forall x \in \dom(f).
\end{align*}

\end{proof}

\subsubsection*{VWF Rounding}
\label{sec:Round1VWF}

In this section, we prove \cref{lemma:Round1VWF} via case analysis.

\begin{proof}[Proof of \cref{lemma:Round1VWF}]
Recall the setting of the lemma:
Given $B_{l, r}$ with $l \le 0$ and $d > 0$, we want to prove that for any $x \ge l$, the following holds:
\begin{align}
\label{eq:Round1VWF}
    dB_{l, r}(x) \ge \frac{dr}{[r]_{1.1}} B_{l, [r]_{1.1}}(x) \ge 2dB_{l, r}\left(\frac{x}{2}\right).
\end{align}

Whether $r$ is positive affects the definition of $B_{l, r}$.
We separate the discussion based on the sign of $r$.
\begin{enumerate}
    \item[\textbf{Case 1.}] $r > 0$: 
        
        Let $u = [r]_{1.1}$ and $c = dr / u$.
        We have that
        \begin{align*}
            dB_{l, r}(x) = \begin{cases}
                \frac{1}{2}dx^2 &, x \in [l, r) \\
                drx - \frac{dr^2}{2} &, x \in [r, \infty),
            \end{cases}
        \end{align*}
        and
        \begin{align*}
            cB_{l, u}(x) = \begin{cases}
                \frac{1}{2}cx^2 &, x \in [l, u) \\
                cux - \frac{cu^2}{2} = drx - \frac{dru}{2} &, x \in [u, \infty).
            \end{cases}
        \end{align*}
        From the definition of $[\cdot]_{1.1}$, we have:
        \begin{align*}
            r \le u \le 1.1r \\
            cu = dr \\
            d \ge c \ge 0.9 d.
        \end{align*}
        
        $dB_{l, r}(x) \ge cB_{l, u}(x)$ holds directly from the relation between derivatives: $dB'_{l, r}(x) \ge cB'_{l, u}(x) \ge 0$.
        
        To prove $cB_{l, u}(x) \ge 2dB_{l, r}(x/2)$, we have to do case analysis on the value of $x$.
        \begin{itemize}
            \item $x \ge 2r$:
            We have
            \begin{align*}
                \left[cB_{l, u}(x) \ge 2dB_{l, r}(x/2)\right]
                \equiv \left[drx - dru/2 \ge drx - dr^2\right]
                \equiv \left[dru/2 \le dr^2\right]
                \equiv \left[u/2 \le r\right]
            \end{align*}
            which is true since $u \le 2r$.
            
            \item $2r \ge x \ge u$:
            We have $r \ge x/2$ and
            \begin{align*}
                \left[cB_{l, u}(x) \ge 2dB_{l, r}(x/2)\right]
                &\equiv \left[drx - dru/2 \ge dx^2 / 4\right] \\
                &\equiv \left[4rx - 2ru/2 \ge x^2\right] \\
                &\equiv \left[4r^2 - 2ru \ge x^2 - 4rx + 4r^2 = (x-2r)^2\right]\\
                &\equiv \left[2r(2r-u) \ge (x-2r)^2\right]
            \end{align*}
            which is true since $|x-2r| \le |u - 2r| \le 2r$.
            
            \item $u \ge x$:
            We have
            \begin{align*}
                \left[cB_{l, u}(x) \ge 2dB_{l, r}(x/2)\right]
                &\equiv \left[cx^2 / 2 \ge dx^2 / 4\right] \\
                &\preceq \left[c \ge d/2\right]
            \end{align*}
            which is true since $c \ge 0.9d > 0$.
        \end{itemize}

    \item[\textbf{Case 2.}] $r \le 0$:
        
        Let $u = [r]_{1.1}$ and $c = dr / u$.
        Since $u, r \le 0$, we have
        \begin{align*}
            dB_{l, r}(x) = \begin{cases}
                \frac{1}{2}d(x-r)^2 &, x \in [l, r) \\
                0 &, x \in [r, \infty),
            \end{cases}
        \end{align*}
        and
        \begin{align*}
            cB_{l, u}(x) = \begin{cases}
                \frac{1}{2}c(x-u)^2 &, x \in [l, u) \\
                0 &, x \in [u, \infty).
            \end{cases}
        \end{align*}
        From the definition of $[\cdot]_{1.1}$ for non-positive numbers, we have:
        \begin{align*}
            1.1r \le u \le r \le 0 \\
            cu = dr \\
            d \ge c \ge 0.9 d.
        \end{align*}
        
        We prove the part $dB_{l, r}(x) \ge cB_{l, u}(x)$ first.
        For $x \ge u$, the inequality holds since RHS is 0 and LHS is always non-negative.
        For $x < u \le r$, we have
        \begin{align*}
            \left[dB_{l, r}(x) \ge cB_{l, u}(x)\right]
            &\equiv \left[d(x-r)^2 / 2 \ge c(x-u)^2 / 2\right]
        \end{align*}
        which is always true since $|x-r| \ge |x-u|$ and $d \ge c > 0$.
        
        Next, we prove the part $cB_{l, u}(x) \ge 2dB_{l, r}(x/2)$.
        For $x/2 > r$, the inequality holds since RHS is 0 and LHS is always non-negative.
        For $x/2 \le r$, we have $x \le 2r \le u$ and
        \begin{align*}
            \left[cB_{l, u}(x) \ge 2dB_{l, r}(x/2)\right]
            &\equiv \left[c(x-u)^2 / 2 \ge d(x/2 - r)^2 = d(x - 2r)^2 / 4\right] \\
            &\equiv \left[c(x-u)^2\ge d(x - 2r)^2/2\right]
        \end{align*}
        which is always true since $|x-2r| \le |x-u|$ and $d/2 \le c$.
\end{enumerate}
\end{proof}

\subsection{Operations on VWF}
\label{sec:VWFDS}

In this section, we prove \cref{lemma:VWFDS} using data structures.

Recall the naive representation of a VWF $f = \{(s_i, r_i, a_i, b_i)\}_{i=0}^{k-1}$ where
\begin{align*}
    f(x) = \frac{1}{2} r_i x^2 + a_i x + b_i, x \in [s_i, s_{i+1}).
\end{align*}
We will use Augmented Binary Search Trees (ABST) to maintain such representation subjects to operations we are interested in.

ABST supports efficient range update and element insertion/deletion.
When $g$ is a general VWF of smaller size, we afford to add pieces of $g$ to $f$ one at a time.
However, adding 1 piece of $g$ to $f$ updates a consecutive segment of $f$'s representation.
Efficient range update is therefore crucial in this setting.

Another reason for using ABST is to support associate range operations corresponds to $\mathrm{Lift}(c)$.
We will show that the updated function
\begin{align*}
    \Bar{f}(x) = \min_{y \in \dom(f)} \frac{1}{2} c(x - y)^2 + f(y)
\end{align*}
has representation $\{(\Bar{s}_i, \Bar{r}_i, \Bar{a}_i, \Bar{b}_i)\}_{i=0}^{k-1}$ where
\begin{align*}
    (\Bar{s}_i, \Bar{r}_i, \Bar{a}_i, \Bar{b}_i) = P_c(s_i, r_i, a_i, b_i), \forall i.
\end{align*}
To maintain the representation in $O(\log n)$-time instead of naive linear-time implementation, we have to exploit the structure of the operator $P_c$ given $c > 0$.
As we shall see, the family of operator $\PP = \{P_c \mid c > 0\}$ is associate and monotone in the sense that $P_a \circ P_b = P_{g(a, b)}$ for some easy-to-compute function $g(\cdot, \cdot)$ and the sequence $\{\Bar{s}_i\}$ is increasing.
Such facts allow the use of ABST and guarantee $O(\log n)$ update time.

The ABST data structure is formalized as following:
\begin{lemma}
\label{lemma:ABST}
Let constant $d > 0$ be the dimensional parameter.
Let $\PP = \{P_c: \Real^{d+1} \to \Real^{d+1}\}_{c \in \Real}$ be a family of associate operators such that (1) for any $a, b \in \Real$, $P_a \circ P_b = P_{g(a, b)}$ for some easy-to-compute function $g(\cdot, \cdot)$ and (2) $P_a$ is an increasing operator in first coordinate.
There is a data structure $\DD$ maintaining a set of key/value pairs $\{(k, \bv)\} \subseteq \Real \times \Real^d$ under the following operations:
\begin{enumerate}
    \item $\mathrm{Insert}((k, \bv))$/$\mathrm{Delete}((k, \bv))$: Insert/Delete the key/value pair $(k, \bv)$ to/from $S$.
    \item $\mathrm{Pred}(y)$: Return the predecessor w.r.t. key $y$, i.e., $\arg\max_{(k, \bv) \in S, k \le y}k$.
    \item $\mathrm{RangeAdd}(\bz \in \Real^d, L, R)$: For every element $(k, \bv)$ of $S$ with its key inside the range $[L, R)$, add its value by $\bz$.
    \item $\mathrm{Op}(c \in \Real)$: Apply operator $P_c$ to every element of $S$.
\end{enumerate}
All above operations can be done in $O(\log n)$-time if $|S|$ is always bounded by $\poly(n)$.
$\DD$ can be initialized in $O(|S|)$-time given the initial set $S$.

Additionally, the data structure is implemented in pointer-machine model and can be maintained persistently with constant overhead.
\end{lemma}

We will use data structure of \cref{lemma:ABST} to maintain VWFs by setting $d=3$.
The family of operators will be defined later.
First, we show the implementation of VWF addition using the data structure.

\begin{claim}
\label{claim:VWFAdd}
Let $f$ (and $g$) be a VWF and $D_f$ (and $D_g$) be the ABST maintaining $f$ ($g$).
If $f$ and $g$ has same domain, one can construct $D_{f + g}$ in $O(\min\{|f|, |g|\}\log n)$-time.
$f+g$ has size at most $|f| + |g|$.
\end{claim}
\begin{proof}
WLOG, we assume that $g$ is smaller.
We initialize $D_{f+g}$ with a snapshot of $D_f$ and insert pieces of $g$ one at a time.

Every piece of $g$ has the following form: $g(x) = (1/2)rx^2+ax + b, x \in [L, R)$.
First, we split pieces of $f$ whose domain overlaps with $[L, R)$.
This can be done via 2 predecessor search and 2 insertion in $D_{f+g}$.

Every piece in $D_{f+g}$ aligns with $[L, R)$, i.e. is either fully contained in or disjoint from $[L, R)$.
Next, for each piece $(s_i, r_i, a_i, b_i)$ of $D_{f+g}$ with $s_i$ lying inside $[L, R)$, we update its coefficients by $(r_i+r, a_i+a, b_i+b)$.
This can be done via one range update in $D_{f+g}$.

Per insertion of a piece of $g$, only a constant number of ABST operations are involved and $O(\log n)$ time spent.
Thus, the total time complexity is $O(|g|\log n)$.

Observe that the set of split points of $f+g$ is actually the union of ones of $f$ and $g$.
Thus, the size of $f+g$ is bounded by the total size of these 2 VWFs.
\end{proof}

The family of operators is $\PP=\{P_c \mid c > 0\}$ where for any $c > 0$, we have
\begin{align*}
    P_c(s, r, a, b) = \left(\frac{c + r}{c}s + \frac{1}{c}a, \frac{c}{c + r}r, \frac{c}{c + r}a, b + \frac{a^2}{2(c+r)}\right).
\end{align*}

\begin{claim}
\label{claim:VWFLift}
Assuming the monotonicity and associativity of $\PP$, given $D_f$, a ABST maintaining $f$, $D_{\Bar{f}}$ can be computed in $O(\log n)$-time where
\begin{align*}
    \Bar{f}(x) = \min_{y \in \dom(f)} \frac{1}{2} c(x-y)^2 + f(y).
\end{align*}
Additionally, $\Bar{f}$ is also a VWF of size at most $|f|+1$.
\end{claim}
\begin{proof}
Let $f=\{(s_i, r_i, a_i, b_i)\}_{i=0}^{k-1}$, we have
\begin{align*}
    f(x) = \frac{1}{2} r_i x^2 + a_i x + b_i, x \in [s_i, s_{i+1}),
\end{align*}
and the sequence of $\{r_i\}$ being non-increasing.
Also, the derivative of $f$ is concave and increasing.

For any $x$, define $y(x) = \arg\max \frac{1}{2} c(x-y)^2 + f(y).$
From optimality condition $y(x)$ is $s_0$ if $cs_0 - cx + f'(s_0) \ge 0$ or some $y$ such that $cy - cx + f'(y) = 0$.
Since $cy + f'(y)$ is increasing in $y$, $y(x)$ is non-decreasing in $x$.

Note that $cy + f'(y) = cy + r_iy + a_i$ for $y \in [s_i, s_{i+1})$ for any $i$.
Since $y(x)$ is non-decreasing, for any $x$ such that $cx \in [cs_i + f'(s_i), cs_{i+1}+f'(s_{i+1}))$,
\begin{align*}
    cx = (c + r_i)y(x) + a_i
\end{align*}
holds and $y(x) = (cx - a_i) / (c + r_i)$.
For $cx < cs_0 + f'(s_0)$, $y(x)$ is $s_0$.

Plug in the formula for $y(x)$, $\Bar{f}(x)$ can be expressed as:
\begin{align*}
    \Bar{f}(x) = \begin{cases}
        \frac{1}{2}\Bar{r}_ix^2 + \Bar{a}_ix + \Bar{b}_i &, x \in [\Bar{s}_i, \Bar{s}_{i+1}) \\
        \frac{1}{2}c(x - s_0)^2 + f(s_0) &, x \in (-\infty, \Bar{s}_0),
    \end{cases}
\end{align*}
where
\begin{align*}
    \Bar{s}_i &= \frac{c + r_i}{c}s_i + \frac{a_i}{c}\\
    \Bar{r}_i &= \frac{c}{c + r_i}r_i \\
    \Bar{a}_i &= \frac{c}{c + r_i}a_i \\
    \Bar{b}_i &= b_i + \frac{a_i^2}{2(c + r_i)}.
\end{align*}

Clearly, $(\Bar{s}_i, \Bar{r}_i, \Bar{a}_i, \Bar{b}_i) = P_c(s_i, r_i, a_i, b_i), \forall i$.
Thus, from the assumption of $\PP$, we can construct $D_{\Bar{f}}$ by (1) apply $\mathrm{Op}(P_c)$ to every element of a snapshot of $D_f$, and then (2) insert a new element corresponds to the case $x \in (-\infty, \Bar{s}_0)$.
The time spent is $O(\log n)$ and the size of $\Bar{f}$ is at most one larger than $f$.

Next, we claim that $\Bar{f}$ is also a VWF.
\begin{itemize}
    \item $\Bar{f}$ is clearly piece-wise quadratic and convex since it is an infimum convolution of 2 convex functions $\frac{1}{2}ct^2$ and $f(t)$.
    \item $\Bar{f}(0) \le 0$ holds since $f(0) \le 0$.
    \item $\Bar{f}'$ is continuous and concave:
    We have
    \begin{align*}
        \Bar{f}'(x) = \begin{cases}
        \Bar{r}_ix + \Bar{a}_i &, x \in [\Bar{s}_i, \Bar{s}_{i+1}) \\
        c(x - s_0) &, x \in (-\infty, \Bar{s}_0).
        \end{cases}
    \end{align*}
    For any $\Bar{s}_i, i > 0$, we have $\lim_{x \to \Bar{s}_i^{-}} \Bar{f}'(x)  = \Bar{r}_{i-1}\Bar{s}_i + \Bar{a}_{i-1} = \Bar{r}_{i}\Bar{s}_i + \Bar{a}_{i} = \Bar{f}'(\Bar{s}_i),$
    where we use the fact that $f'$ is continuous, i.e, $r_{i-1}s_i+a_i = r_is_i + a_i$ holds.
    
    To see that $\Bar{f}'$ is concave, observe that the sequence $\{\Bar{r}_i\}$ is non-increasing and is upper-bounded by $c$.
\end{itemize}

\end{proof}

Next, we show that $\PP$ is associate and commute.
\begin{claim}
For any $x, y > 0$, we have $P_x \circ P_y = P_{g(x, y)}$ where $g(x, y) = (x^{-1}+y^{-1})^{-1}.$
\end{claim}
\begin{proof}
This can be proved via direct evaluation.
Let $z = g(x, y) = xy/(x+y)$.
Given any $(s, r, a, b)$ with $r > 0$, we have
\begin{align*}
    P_x \circ P_y(s, r, a, b) 
    &= P_x\left(\frac{y+r}{y}s + \frac{1}{y}a, \frac{y}{y+r}r, \frac{y}{y+r}a, b + \frac{a^2}{2(y + r)}\right) \\
    &= \left(\Bar{s}, \Bar{r}, \Bar{a}, \Bar{b}\right),
\end{align*}
where
\begin{align*}
    \Bar{s} 
    &= \frac{1}{x}\left[\left(x + \frac{y}{y+r}r\right)\left(\frac{y+r}{y}s + \frac{1}{y}a\right)+\frac{y}{y+r}a\right] \\
    &= \left(\frac{y+r}{y} + \frac{r}{x}\right)s + \left(\frac{1}{y} + \frac{r}{x(y+r)} + \frac{y}{x(y+r)}\right)a \\
    &= \frac{xy+(x+y)r}{xy}s + \left(\frac{1}{x}+\frac{1}{y}\right)a = \frac{z + r}{z}s + \frac{1}{z}a,
\end{align*}
\begin{align*}
    \Bar{r}
    = \frac{x}{x + \frac{y}{y+r}r}\frac{y}{y+r}r
    = \frac{xy}{x(y+r) + yr} r
    = \frac{z}{z + r}r,
\end{align*}
\begin{align*}
    \Bar{a}
    = \frac{x}{x + \frac{y}{y+r}r}\frac{y}{y+r}a
    = \frac{xy}{x(y+r) + yr} a
    = \frac{z}{z + r}r,
\end{align*}
and
\begin{align*}
    \Bar{b}
    &= b + \frac{a^2}{2(y + r)} + \frac{1}{2(x + \frac{y}{y+r}r)}\left(\frac{y}{y+r}a\right)^2 \\
    &= b + \left(\frac{1}{2(y+r)} + \frac{1}{2(x + \frac{y}{y+r}r)}\frac{y^2}{(y+r)^2}\right)a^2 \\
    &= b + \left(\frac{y^2+x(y+r)+yr}{2(x(y+r)+yr)(y+r)}\right)a^2 \\
    &= b + \left(\frac{y+x}{2(x(y+r)+yr)}\right)a^2 = b + \frac{1}{2(z+r)}a^2.
\end{align*}

Therefore, we have $P_x \circ P_y(s, r, a, b) = P_z(s, r, a, b)$.
\end{proof}

\begin{proof}[Proof of \cref{lemma:VWFDS}]
From discussions above, we know how to maintain the desired data structure using ABST from \cref{lemma:ABST}.
\end{proof}

\section{Inexact Proximal Accelerated Gradient Descent}
\label{sec:ProxAGD}


In this section, we prove \cref{lemma:ProxAGD}.

Proximal Gradient method is a family of algorithms that solves problems of form:
\begin{align}
\label{eq:CompCvx}
    \min_{\bx \in K} \Phi(\bx) \coloneqq f(\bx) + h(\bx),
\end{align}
where $K$ is a closed convex set, $f$ is a smooth convex function, and $h$ is a proper convex function.
Traditional gradient methods can only solves smooth convex problems.
For the non-smooth case, one needs to leverage the structure of the problem.
In \cite{BeckT09, Tseng08, Nesterov13}, they show one can apply gradient methods, even accelerated ones, in solving Problem~(\ref{eq:CompCvx}).

The diffusion problems have an identical structure with Problem~(\ref{eq:CompCvx}).
Given an instance $\GG$, one wants to solve:
\begin{align}
\label{eq:DiffCvx}
    \min_{\bx \ge \bb} \EE^\GG(\bx) \coloneqq \frac{1}{2}\bx^\top\bL(G)\bx + \sum_u f_u(x_u). 
\end{align}
One can see the identity holds when we take $K=\{\bx \ge \bb\}$, $f$ being $\frac{1}{2}\bx^\top\bL(G)\bx$, and $h$ being the sum of VWFs.
Given some spectral sparsifier $H$ of $G$, $f$ is smooth w.r.t. the matrix norm induced by $\bL(H)$.

However, we cannot apply these well-known results directly.
Proximal gradient methods heavily rely on computing \emph{proximal mapping}, i.e. solving the sub-problem given $\bx \in K$:
\begin{align}
\label{eq:ProxEx}
    \min_{\by \in K} \grad f(\bx)^\top(\by - \bx) + \frac{L}{2}\norm{\by - \bx}^2 + h(\by).
\end{align}
In the context of diffusion, such sub-problem is yet another diffusion problem on a possibly smaller instance.
The algorithms for diffusion problems achieve only high accuracy.

\cite{SRB11} and \cite{VSBV13} showed that proximal gradient methods still work even if we only solve Problem~(\ref{eq:ProxEx}) approximately.
But they only prove results when the smoothness is measured w.r.t. the Euclidean norm.

In this section, we will generalize their results and apply them to solve diffusion problems.
We start by introducing some fundamental notions.

\subsection{Preliminaries}

\paragraph{$\eps$-subdifferential}
The \emph{$\eps$-subdifferential} of any convex function $a$ at $\bx$ is the set of vectors $\bv$ such that for any $\bu$, $a(\bu) \ge a(\bx) + \bv^\top (\bu - \bx) - \eps$ holds.
It is denoted as $\partial_\eps a(\bx)$.
The \emph{convex conjugate} of a function $a$ is a convex function defined as $a^*(\bv) = \sup_\bu \bv^\top\bu - a(\bu).$

\begin{fact}
Given any $\bv$, $\bv$ is in the $\eps$-subdifferential of function $a$ at $\bx$ if and only if $a^*(\bv) \le \eps + \bv^\top\bx - a(\bx).$
\end{fact}
\begin{fact}
$\partial_\eps (a_1 + a_2)(\bx) \subseteq \partial_\eps a_1(\bx) + \partial_\eps a_2(\bx)$
\end{fact}


\paragraph{Matrix Norm}
In this section, we consider the norm induced by some psd matrix $\bA$.
Therefore, $\norm{\bx} = \sqrt{\bx^\top\bA\bx}$ and $\norm{\bx}_* = \sqrt{\bx^\top\bA^\dagger\bx}.$
We also define $R(x) = \norm{\bx}^2$ with $\grad R(x) = \bA \bx.$

\paragraph{Smoothness and Strong Convexity}
A convex differentiable function $f$ is \emph{$L$-smooth} w.r.t. some (semi)norm $\norm{\cdot}$ if $f(\by) \le f(\bx) + \grad f(\bx)^\top (\by - \bx) + (L/2)\norm{\by - \bx}^2$ for any $\by, \bx \in \dom(f)$.
$f$ is \emph{$\mu$-strongly-convex} if $f(\by) \ge f(\bx) + \grad f(\bx)^\top (\by - \bx) + (\mu/2)\norm{\by - \bx}^2$ for any $\by, \bx \in \dom(f)$.

Given a $L$-smooth function $f$ and any $\bx \in \dom(f)$, we define the \emph{quadratic approximation of $f$ around $\bx$} as $\Bar{f}_\bx(\bu) = f(\bx) + \grad f(\bx)^\top (\bu - \bx) + (L/2)\norm{\bu - \bx}^2.$




\paragraph{Objective Function}
Our goal is to find (approximate) optimal solution to the problem
\begin{align*}
    \min_{\bx} \Phi(\bx) \coloneqq f(\bx) + h(\bx),
\end{align*}
where $f$ is $L$-smooth and $\mu$-strongly-convex w.r.t norm $\norm{\cdot}$ and $h$ is a lower semi-continuous proper convex function.

\paragraph{(Inexact) Proximal Operator}
Given any $\bx$, we define the \emph{Proximal Operator} on $\bx$ as
\begin{align*}
    p_\bx \in \arg\min_{\by} \Bar{\Phi}_\bx(\by) \coloneqq \Bar{f}_\bx(\by) + h(\by),
\end{align*}
the resulting $p_\bx$ is called a \emph{Proximal Mapping} of $\bx$.

However, finding $p_\bx$ may be not easier.
In our application, we can only compute \emph{Inexact Proximal Mapping} of $\bx$, say $q_\bx$ such that
\begin{align*}
    \Bar{\Phi}_\bx(q_\bx) - \Bar{\Phi}_\bx(p_\bx) \le \eps.
\end{align*}




\subsection{Meta Algorithm}
\label{sec:MetaAlg}

\begin{algorithm}
\caption{Accelerated Method using inexact gradient mapping}
\label{algo:ProxAGD}
\begin{algorithmic}[1]
    \Procedure{ProxAGD}{$\bx^{(0)}, T$}
        \State $\by^{(0)}, \bz^{(0)} \coloneqq \bx^{(0)}$
        \For{$k=0, 1, \ldots, T - 1$}
            \State $\alpha_{k+1} \coloneqq (k+2) / 2$
            \State $\tau_k \coloneqq 1 / (\alpha_{k+1}) = 2 / (k + 2)$
            \State $\bx^{(k+1)} \coloneqq (1 - \tau_k)\by^{(k)} + \tau_k \bz^{(k)}$
            \State Compute $q_{k+1} \approx \arg\min_\by \Bar{\Phi}_\bx(\by)$ with additive error $\eps_{k+1}$
            \State $\by^{(k+1)} \coloneqq q_{k+1}$
            \State $\bz^{(k + 1)} \coloneqq \bz^{(k)} + \alpha_{k+1}(\by^{(k+1)} - \bx^{(k+1)})$
        \EndFor
        \State \Return $\by^{(T)}$
    \EndProcedure
\end{algorithmic}
\end{algorithm}

In this section, we present an meta algorithm called\emph{Inexact Proximal Accelerated Gradient Descent (Inexact ProxAGD)} and analyze its convergence.
The pseudocode is presented as \cref{algo:ProxAGD}.

For simplicity, we assume the smooth part of the objective, $f$, is $1$-smooth.
The convergence result is as follows:
\begin{theorem}
\label{thm:InexactAPGD}
In \cref{algo:ProxAGD}, the following holds for any $T > 1$:
\begin{align*}
    \Phi(\by^{(T)}) - \Phi(\bx^*)
    &\le \frac{2}{(T+1)^2}\left(\norm{\bx^{(0)} - \bx^*} + 3\sum_{k=1}^{T}k \sqrt{2\eps_{k}}\right)^2
\end{align*}
\end{theorem}

The rest of the section proves \cref{thm:InexactAPGD}.

\subsubsection*{Structural Facts On Inexact Proximal Mappings}

Given any $\bx$, let $q_\bx$ be an approximation to $p_\bx$ with additive error $\eps_\bx$.
That is, $\Bar{\Phi}_\bx(q_\bx) \le \eps_\bx + \Bar{\Phi}_\bx(p_\bx).$
This section presents some structural facts about $q_\bx$.
These facts play important roles in analysis of the meta algorithm.

\begin{fact}
$0 \in \partial_{\eps_\bx} \Bar{\Phi}(q_\bx).$
\end{fact}
\begin{proof}
    Observe that for any $\bu$, we have $\Bar{\Phi}(\bu) \ge \Bar{\Phi}(p_\bx) \ge \Bar{\Phi}(q_\bx) + 0^\top(\bu - q_\bx) - \eps_\bx.$
\end{proof}

\begin{lemma}[$\eps$-subdifferential of $h(q_\bx)$]
\label{lemma:subdiff_h}
There exists some $\bs_\bx$ such that
\begin{enumerate}
    \item $\norm{\bs_\bx}_*^2 \le 2\eps_\bx$
    \item $-\bs_\bx - \grad \Bar{f}_\bx(q_\bx) \in \partial_{\eps_\bx}h(q_\bx).$
\end{enumerate}
\end{lemma}
\begin{proof}
For any $\bv$, $\bv \in \partial_{\eps_\bx}\Bar{f}(q_\bx)$ if $\eps_\bx + \bv^\top q_\bx - \Bar{f}_\bx(q_\bx) \ge \Bar{f}_\bx^*(\bv).$

From the definition of dual norm, we have:
\begin{align*}
    \Bar{f}_\bx^*(\bv)
    &= \sup_\bu \bv^\top\bu - \Bar{f}_\bx(\bu) \\
    &= \sup_\bu \bv^\top\bu - f(\bx) - \grad f(\bx)^\top (\bu - \bx) - \frac{1}{2}\norm{\bu - \bx}^2 \\
    &= -f(x) + \bv^\top\bx + \sup_\bu (\bv - \grad f(\bx))^\top (\bu - \bx) - \frac{1}{2}\norm{\bu - \bx}^2 \\
    &= -f(x) + \bv^\top\bx + \frac{1}{2}\norm{\bv - \grad f(\bx)}_*^2
\end{align*}

Therefore, we have
\begin{align*}
    \eps_\bx 
    &\ge \Bar{f}_\bx^*(\bv) + \Bar{f}_\bx(q_\bx) - \bv^\top q_\bx \\
    &= -f(x) + \bv^\top\bx + \frac{1}{2}\norm{\bv - \grad f(\bx)}_*^2 + f(x) + \grad f(x)^\top(q_\bx - \bx) + \frac{1}{2}\norm{q_\bx - \bx}^2 - \bv^\top q_\bx \\
    &= (\bv - \grad f(\bx))^\top(\bx - q_\bx) + \frac{1}{2}\norm{\bv - \grad f(\bx)}_*^2 + \frac{1}{2}\norm{q_\bx - \bx}^2 \\
    &= \frac{1}{2}\norm{\bv - \grad f(\bx) - \bA(q_\bx - \bx)}_*^2 = \frac{1}{2}\norm{\bv - \grad \Bar{f}_\bx(q_\bx)}_*^2
\end{align*}

Let $\bs = \bv - \grad \Bar{f}_\bx(q_\bx)$.
Since $0 \in \partial_{\eps_\bx}\Bar{\Phi}_\bx(q_\bx)$ and $\bs + \grad \Bar{f}_\bx(q_\bx) \in \partial_{\eps_\bx}\Bar{f}_\bx(q_\bx)$, we have
\begin{align*}
    -\bs - \grad \Bar{f}_\bx(q_\bx) \in \partial_{\eps_\bx}h(q_\bx).
\end{align*}
\end{proof}

The following lemma analyzes the difference between $\Phi(q_\bx)$ and $\Phi(\bu)$ geometrically.
\begin{lemma}[Inexact Gradient Lemma]
\label{lemma:inexact_grad_lemma}
Given any $\bx, \bu$,
\begin{align*}
    \Phi(\bu)
    &\ge \Phi(q_\bx) + \bd_\bx^\top(\bu - \bx) + \frac{1}{2\kappa}\norm{\bu - \bx}^2 + \frac{1}{2}\norm{\bd_\bx}_*^2 - \sqrt{2\eps_\bx} \norm{\bu - q_\bx} - \eps_\bx \\
    &\ge \Phi(q_\bx) + \bd_\bx^\top(\bu - \bx) + \frac{1}{2}\norm{\bd_\bx}_*^2 - \sqrt{2\eps_\bx} \norm{\bu - q_\bx} - \eps_\bx,
\end{align*}
holds where $\bd_\bx \coloneqq \bA(\bx - q_\bx).$
\end{lemma}
\begin{proof}
From Lemma~\ref{lemma:subdiff_h}, there's $\bs_\bx$ with $\norm{\bs_\bx}_*^2 \le 2\eps_\bx$ such that for any $\bu$,
\begin{align*}
    h(\bu) + \eps_\bx 
    &\ge h(q_\bx) + (-\bs_\bx - \grad \Bar{f}_\bx(q_\bx))^\top(\bu - q_\bx)\\
    &= h(q_\bx) + (-\bs_\bx - \grad f(\bx) - \bA(q_\bx - \bx))^\top (\bu - q_\bx) \\
    &= h(q_\bx) + \bs_\bx^\top (q_\bx - \bu) - \grad f(\bx) ^\top (\bu - q_\bx) + \bd_\bx^\top(\bu - q_\bx) \\
    &\ge h(q_\bx) - \sqrt{2\eps_\bx} \norm{q_\bx - \bu} - \grad f(\bx) ^\top (\bu - q_\bx) + \bd_\bx^\top(\bu - q_\bx).
\end{align*}

Since $f$ is 1-smooth, we have
\begin{align*}
    f(q_\bx)
    \le f(\bx) + \grad f(\bx)^\top (q_\bx - \bx) + \frac{1}{2}\norm{q_\bx - \bx}^2
    = f(\bx) + \grad f(\bx)^\top (q_\bx - \bx) + \frac{1}{2}\norm{\bd_\bx}_*^2
\end{align*}
Using convexity of $f$, we have
\begin{align*}
    f(\bu) &\ge f(\bx) + \grad f(\bx)^\top(\bu - \bx) + \frac{1}{2\kappa}\norm{\bu - \bx}^2 \\
    &\ge \left(f(q_\bx) - \grad f(\bx)^\top(q_\bx - \bx) - \frac{1}{2}\norm{\bd_\bx}_*^2\right) + \grad f(\bx)^\top(\bu - \bx) + \frac{1}{2\kappa}\norm{\bu - \bx}^2\\
    &= f(q_\bx) + \grad f(\bx)^\top (\bu - q_\bx) - \frac{1}{2}\norm{\bd_\bx}_*^2 + \frac{1}{2\kappa}\norm{\bu - \bx}^2.
\end{align*}

Sum up the inequality regard $f(\bu)$ and $h(\bu)$, we have
\begin{align*}
    \Phi(\bu) + \eps_\bx
    \ge \Phi(q_\bx) - \sqrt{2\eps_\bx} \norm{q_\bx - \bu} + \bd_\bx^\top(\bu - q_\bx) - \frac{1}{2}\norm{\bd_\bx}_*^2 + \frac{1}{2\kappa}\norm{\bu - \bx}^2.
\end{align*}

We conclude the proof by observing that:
\begin{align*}
    \bd_\bx^\top(\bu - q_\bx) = \bd_\bx^\top(\bu - \bx) + \bd_\bx^\top(\bx - q_\bx) = \bd_\bx^\top(\bu - \bx) + \norm{\bd_\bx}_*^2.
\end{align*}

\end{proof}

The following lemma analyze the term $\bd_\bx^\top(\bu - \bx)$:
\begin{lemma}[Mirror Descent Lemma]
\label{lemma:exact_mirror_lemma}
For any $\bx, \bd, \bu$ and $\eta > 0$, let $\bz = \bx - \eta \bA^\dagger \bd$, we have
\begin{align*}
    \bd^\top (\bx - \bu) 
    &= \frac{\eta}{2}\norm{\bd}_*^2 + \frac{1}{2\eta}\left(\norm{\bx - \bu}^2 - \norm{\bz - \bu}^2\right)
\end{align*}
\end{lemma}
\begin{proof}
The idea of the proof is the generalized Pythagorean theorem:
\begin{align*}
    \eta \bd^\top (\bx - \bu)
    &= (\bx - \bz)^\top \bA (\bx - \bu) \\
    &= \frac{1}{2} (\bx - \bz)^\top \bA (\bx - \bz) + \frac{1}{2} (\bx - \bu)^\top \bA (\bx - \bu) - \frac{1}{2} (\bz - \bu)^\top \bA (\bz - \bu)
\end{align*}
\end{proof}

Given these 2 lemmas, we are ready to analyze \cref{algo:ProxAGD}.
Instead of letting $\bu = \bx^*$ in \cref{lemma:inexact_grad_lemma}, we plug the following:
\begin{align*}
    \tau_k \bx^* + (1-\tau_k) \by^{(k)}.
\end{align*}

The following facts are important in simplifying intermediate formulas:
\begin{fact}
\label{fact:f1}
$\frac{1-\tau_k}{\tau_k} = \frac{1 - 2/(k+2)}{2/(k+2)} = \frac{k}{2} = \alpha_{k+1}-1.$
\end{fact}
\begin{fact}
\label{fact:f2}
$\tau_k \bu + (1-\tau_k) \by^{(k)} - \bx^{(k+1)} = \tau_k (\bu - \bz^{(k)})$
\end{fact}
\begin{fact}
\label{fact:f3}
$\tau_k \bu + (1-\tau_k) \by^{(k)} - \by^{(k+1)} = \tau_k (\bu - \bz^{(k+1)})$
\end{fact}
\begin{definition}
For any $k \ge 0$, define $\bd_k \coloneqq \bA(\bx^{k} - \by^{k})$.
\end{definition}
\begin{definition}
For any $k \ge 0$, define $D_k \coloneqq \norm{\bx^* - \bz^{(k)}}$.
\end{definition}
\begin{fact}
\label{fact:f4}
$\bz^{(k+1)} = \bz^{(k)} + \alpha_{k+1}(\by^{(k+1)} - \bx^{(k+1)}) = \bz^{(k)} - \alpha_{k+1}\bA^\dagger \bd_{k+1}$
\end{fact}
\begin{fact}
\label{fact:f5}
By \cref{lemma:exact_mirror_lemma} and \cref{fact:f5}, we have
$\bd_{k+1}^\top(\bz^{(k)} - \bx^*)= \frac{\alpha_{k+1}}{2}\norm{\bd_{k+1}}_*^2 + \frac{1}{2\alpha_{k+1}}(D_k^2 - D_{k+1}^2).$
\end{fact}
\begin{fact}
\label{fact:f6}
$\alpha_{k+1}^2-(\alpha_{k+2}^2-\alpha_{k+2}) = (k+2)^2/4 - (k+3)^2/4 + (k+3)/2 = 1/4$.
\end{fact}
\begin{fact}
\label{fact:f7}
$\sum_{k=0}^{T-1}\alpha_{k+1} = \sum_{k=0}^{T-1}(k+2)/2 = T(T+3)/4$.
\end{fact}

By convexity, we have
\begin{align*}
    \tau_k \Phi(\bx^*) + (1 - \tau_k)\Phi(\by^{(k)})
    \ge &\Phi\left(\tau_k \bx^* + (1-\tau_k) \by^{(k)}\right) \\
    \ge &\Phi(\by^{(k+1)}) + \bd_{k+1}^\top(\tau_k \bx^* + (1-\tau_k) \by^{(k)} - \bx^{(k+1)}) + \frac{1}{2}\norm{\bd_{k+1}}_*^2 \\
    &- \sqrt{2\eps_{k+1}} \norm{\tau_k \bx^* + (1-\tau_k) \by^{(k)} - \by^{(k+1)}} - \eps_{k+1} \\
    = &\Phi(\by^{(k+1)}) + \tau_k\bd_{k+1}^\top(\bx^* - \bz^{(k)}) + \frac{1}{2}\norm{\bd_{k+1}}_*^2 \\
    &- \tau_k\sqrt{2\eps_{k+1}} D_{k+1} - \eps_{k+1} \\
    = &\Phi(\by^{(k+1)}) - \frac{\tau_k\alpha_{k+1}}{2}\norm{\bd_{k+1}}_*^2 - \frac{\tau_k}{2\alpha_{k+1}}(D_k^2 - D_{k+1}^2) + \frac{1}{2}\norm{\bd_{k+1}}_*^2 \\
    &- \tau_k\sqrt{2\eps_{k+1}} D_{k+1} - \eps_{k+1} \\
    = &\Phi(\by^{(k+1)}) - \frac{\tau_k^2}{2}(D_k^2 - D_{k+1}^2) - \tau_k\sqrt{2\eps_{k+1}} D_{k+1} - \eps_{k+1},
    \addtocounter{equation}{1}\tag{\theequation} \label{eq:AGDMainIneq}
\end{align*}
where the second inequality uses \cref{lemma:inexact_grad_lemma}, the first equality uses \cref{fact:f2} and \cref{fact:f3}, the second equality uses \cref{fact:f5}, and the third equality uses that $\tau_k\alpha_{k+1} = 1$.

Multiplying both side with $\alpha_{k+1}^2=\tau_k^{-2}$ and then rearranging terms by \cref{fact:f1} yields
\begin{align*}
    \alpha_{k+1}^2\Phi(\by^{(k+1)}) - (\alpha_{k+1}^2 - \alpha_{k+1})\Phi(\by^{(k)}) 
    \le &\alpha_{k+1}\Phi(\bx^*) + \frac{1}{2}(D_k^2 - D_{k+1}^2) \\
    &+ \alpha_{k+1}\sqrt{2\eps_{k+1}} D_{k+1} + \alpha_{k+1}^2 \eps_{k+1}.
    \addtocounter{equation}{1}\tag{\theequation} \label{eq:AGD_per_step}
\end{align*}

Summing over $k=0, 1, \ldots, T-1$ yields
\begin{align*}
    \sum_{k=0}^{T-1} \alpha_{k+1}^2\Phi(\by^{(k+1)}) - (\alpha_{k+1}^2 - \alpha_{k+1})\Phi(\by^{(k)})
    = &\alpha_T^2\Phi(\by^{(T)}) + \sum_{k=1}^{T-1}\frac{1}{4}\Phi(\by^{(k)}) \\
    \le &\sum_{k=0}^{T-1} \alpha_{k+1}\Phi(\bx^*) + \frac{1}{2}(D_0^2 - D_T^2) \\
    &+ \sum_{k=0}^{T-1} \alpha_{k+1}\sqrt{2\eps_{k+1}} D_{k+1} + \alpha_{k+1}^2 \eps_{k+1}
\end{align*}
where the first equality uses telescoping and \cref{fact:f6}, and the first inequality uses Inequality~(\ref{eq:AGD_per_step}) and telescoping.

Subtracting both side by $\sum_{k=1}^{T-1}\frac{1}{4}\Phi(\by^{(k)})$ and then using the optimality of $\bx^*$ yields
\begin{align*}
    \alpha_T^2\Phi(\by^{(T)}) \le \left(\sum_{k=0}^{T-1} \alpha_{k+1} - \frac{T-1}{4}\right)\Phi(\bx^*) + \frac{1}{2}(D_0^2 - D_T^2) + \sum_{k=0}^{T-1} \alpha_{k+1}\sqrt{2\eps_{k+1}} D_{k+1} + \alpha_{k+1}^2 \eps_{k+1}.
\end{align*}
The definition of $\alpha_k$ and \cref{fact:f7} implies that $\alpha_T^2 = \sum_{k=0}^{T-1} \alpha_{k+1} - \frac{T-1}{4} = (T+1)^2/4$.
Thus, dividing both side by $(T+1)^2/4$ and then subtracting $\Phi(\bx^*)$ from both sides yields
\begin{align}
\label{eq:AGDPreResult}
    \Phi(\by^{(T)}) - \Phi(\bx^*) \le \frac{4}{(T+1)^2}\left(\frac{1}{2}(D_0^2 - D_T^2) + \sum_{k=0}^{T-1} \alpha_{k+1}\sqrt{2\eps_{k+1}} D_{k+1} + \alpha_{k+1}^2 \eps_{k+1}\right).
\end{align}

The following fact from \cite{SRB11} helps further simplifying the inequality.
\begin{fact}[Lemma 1, \cite{SRB11}]
\label{fact:RecSRB11}
Assume that an nonnegative sequence $\{u_k\}$ satisfies the following recursion for all $k \ge 0$:
\begin{align*}
    u_k^2 \le S_k + \sum_{i=1}^{k} \lambda_i u_i,
\end{align*}
with $\{S_k\}$ an increasing sequence and $\lambda_i \ge 0, \forall i$, the following holds for all $k \ge 1$:
\begin{align*}
    u_k \le \frac{1}{2}\sum_{i=1}^k \lambda_i + \left(S_k + \left( \frac{1}{2}\sum_{i=1}^{k} \lambda_i\right)^2\right)^{1/2}
\end{align*}
\end{fact}

Inequality~(\ref{eq:AGDPreResult}) and the optimality of $\bx^*$ yields
\begin{align*}
    D_T^2 
    &\le D_0^2 + 2\sum_{k=1}^{T} \alpha_{k}^2 \eps_{k} + 2\sum_{k=1}^{T} \alpha_{k}\sqrt{2\eps_{k}} D_{k} \\
    &\le D_0^2 + \sum_{k=1}^{T} 2k^2 \eps_{k} + \sum_{k=1}^{T} 2k\sqrt{2\eps_{k}} D_{k},
\end{align*}
where the second inequality uses that $\alpha_k \le k$.

\cref{fact:RecSRB11} (using $S_k = D_0^2 + \sum_{i=1}^{k} 2i^2 \eps_i$ and $\lambda_i = 2i\sqrt{2\eps_i}$) and denoting $A_T = \sum_{k=1}^{T}k \sqrt{2\eps_{k}}$, $B_T = \sum_{k=1}^{T}k^2\eps_{k}$ yields
\begin{align*}
    D_T 
    &\le \frac{1}{2}\sum_{k=1}^T 2k\sqrt{2\eps_k} + \left(D_0^2 + \sum_{k=1}^{T} 2k^2 \eps_{k} + \left(\frac{1}{2}\sum_{k=1}^T 2k\sqrt{2\eps_k}\right)^2\right) \\
    &= A_T + \left(D_0^2 + 2B_T + A_T^2\right)^{1/2} \\
    &\le A_T + D_0 + \sqrt{2B_T} + A_T \\
    &\le D_0 + 3A_T
\end{align*}
where the second inequality uses Cauchy-Schwarz, and the third inequality uses $A_T^2 \ge 2B_T$.

Combining with Inequality~(\ref{eq:AGDPreResult}) yields
\begin{align*}
    \Phi(\by^{(T)}) - \Phi(\bx^*) 
    &\le \frac{4}{(T+1)^2}\left(\frac{1}{2}D_0^2 + \sum_{k=1}^{T} \alpha_{k}\sqrt{2\eps_{k}} (D_0 + 3A_{k}) + \sum_{k=1}^{T}\alpha_{k}^2 \eps_{k}\right) \\
    &\le \frac{4}{(T+1)^2}\left(\frac{1}{2}D_0^2 + \sum_{k=1}^{T} k\sqrt{2\eps_{k}} (D_0 + 3A_T) + \sum_{k=1}^{T}k^2 \eps_{k}\right) \\
    &= \frac{4}{(T+1)^2}\left(\frac{1}{2}D_0^2 + A_T (D_0 + 3A_T) + B_T\right) \\
    &\le \frac{4}{(T+1)^2}\left(\frac{1}{2}D_0^2 + A_T (D_0 + 3A_T) + \frac{1}{2}A_T^2\right) \\
    &\le \frac{2}{(T+1)^2}\left(D_0+3A_T\right)^2 \\
    &= \frac{2}{(T+1)^2}\left(\norm{\bx^{(0)} - \bx^*} + 3\sum_{k=1}^{T}k \sqrt{2\eps_{k}}\right)^2,
\end{align*}
where the first inequality also uses non-negativity of $D_T$, the second inequality uses monotonicity of $\{A_k\}$, the first equality uses definition of $A_T$ and $B_T$, and the third inequality uses $A_T^2 \ge 2B_T$.

\begin{proof}[Proof of \cref{thm:InexactAPGD}]
Discussions above proved the convergence result.
\end{proof}

\subsection{Application in Diffusion Problems}
\label{sec:AppDiffusion}

In this section, we apply \cref{thm:InexactAPGD} to prove \cref{lemma:ProxAGD}.

Specifically, we show that diffusion problem fits the problem structure solved by \cref{algo:ProxAGD} and the inexact proximal mapping can be computed efficiently.
In \cref{lemma:ProxAGD}, we are given a diffusion instance $\GG$ and $H$, a $\kappa$-spectral sparsifier of $G$.
The diffusion problem induced by $\GG$ is to find feasible potential whose energy is close to
\begin{align*}
    \min_{\bx \ge \bb^\GG} \EE^\GG(\bx) = \frac{1}{2}\norm{\bx}^2_{\bL(G)} + \sum_u f_u(x_u).
\end{align*}

Define $K = \{\bx \ge \bb^\GG\}$, $g(\bx) = \frac{1}{2}\norm{\bx}^2_{\bL(G)}$, and $h(\bx) = \sum_u f_u(x_u) + I_K(\bx)$ where $I_K(\bx) = 0$ if $\bx \in K$ or $\infty$ otherwise.
$h(\bx)$ is a continuous proper convex function.
Since $H$ is a $\kappa$-spectral sparsifier of $G$, the following claim holds:
\begin{claim}
$g(\bx)$ is 1-smooth and $\kappa^{-1}$-strongly-convex w.r.t. $\norm{\cdot}_{\bL(H)}$, the norm induced by the Laplacian of $H$.
\end{claim}
\begin{proof}
Let $B_g(\by; \bx)$ be $g$'s Bregman Divergence.
The definition of $g$ yields:
\begin{align*}
    B_g(\by; \bx) 
    &= g(\by) - g(\bx) - \grad g(\bx)^\top(\by - \bx) \\
    &= \frac{1}{2}\norm{\by}^2_{\bL(G)} - \frac{1}{2}\norm{\bx}^2_{\bL(G)} - \left(\bL(G)\bx\right)^\top(\by - \bx) \\
    &= \frac{1}{2}\norm{\by - \bx}^2_{\bL(G)}.
\end{align*}
The property of $H$ yields:
\begin{align*}
    \frac{1}{2\kappa}\norm{\by - \bx}^2_{\bL(H)} \le B_g(\by; \bx) = \frac{1}{2}\norm{\by - \bx}^2_{\bL(G)} \le \frac{1}{2}\norm{\by - \bx}^2_{\bL(H)}.
\end{align*}
Therefore, $g$ is 1-smooth and $\kappa^{-1}$-strongly-convex w.r.t. $\norm{\cdot}_{\bL(H)}$.
\end{proof}

Define $\Phi(\bx) = g(\bx) + h(\bx)$, the diffusion problem has the following form:
\begin{align*}
    \min_{\bx}\Phi(\bx) = g(\bx) + h(\bx),
\end{align*}
which can be approximately solved by \cref{algo:ProxAGD}.
Note that the constraint on potentials is included in the function $h$.

Each step of \cref{algo:ProxAGD} requires computing inexact proximal mapping.
Given any potential $\bx$ which can be infeasible, $p_\bx$, the proximal mapping of $\bx$, is the optimal solution to the following problem:
\begin{align*}
    p_\bx = \arg\min_{\by} \left\{\Bar{\Phi}_\bx(\by) = \Bar{g}_\bx(\by) + h(\by)\right\}.
\end{align*}
Instead of the exact proximal mapping $p_\bx$, an approximate solution $q_\bx$ is acceptable if $\Bar{\Phi}_\bx(q_\bx) - \Bar{\Phi}_\bx(p_\bx)$ is small enough.

Finding $q_\bx$ is another diffusion problem.
Using that $\Bar{g}_\bx$ is the quadratic approximation of 1-smooth function $g$ around $\bx$ w.r.t. the norm $\norm{\cdot}_{\bL(H)}$ and the definition of function $h$ yields 
\begin{align*}
    \Bar{\Phi}_\bx(\by) &= \Bar{g}_\bx(\by) + h(\by) = \left(g(\bx) + \grad g(\bx)^\top(\by - \bx) + \frac{1}{2}\norm{\by - \bx}_{\bL(H)}^2\right) + \sum_{u}f_u(y_u) + I_K(\by) \\
    &= \frac{1}{2}\norm{\bx}_{\bL(G)}^2 + \bx^\top \bL(G) (\by - \bx) + \frac{1}{2}\norm{\by - \bx}_{\bL(H)}^2 + \sum_{u}f_u(y_u) + I_K(\by).
\end{align*}
By shifting the objective by a constant $\Bar{\Phi}_\bx(\mb{0})$, computing exact proximal mapping is equivalent to solving the following:
\begin{align*}
    \min_{\by \in K} \Bar{\Phi}_\bx(\by) - \Bar{\Phi}_\bx(\mb{0}) = \frac{1}{2}\norm{\by}_{\bL(H)}^2 + \sum_u\left[\left(\bL(G)\bx - \bL(H)\bx\right)_uy_u + f_u(y_u) - f_u(0)\right],
\end{align*}
where we use the fact that $\mb{0} \in K$.

For each vertex $u$, define $f^\bx_u(t) = \left(\bL(G)\bx - \bL(H)\bx\right)_ut + f_u(t) - f_u(0)$.
Such $f^\bx_u(t)$ is a VWF of same size as $f_u(x)$.
The problem above can be written as
\begin{align}
\label{eq:ProxDiff}
    \min_{\by \ge \bb^\GG} \frac{1}{2}\norm{\by}_{\bL(H)}^2 + \sum_u f^\bx_u(y_u),
\end{align}
which is a diffusion problem.
We summarize the discussion as the following claim:
\begin{claim}
\label{claim:ProxDiff}
Given any potential $\bx$ for some instance $\GG$, define the diffusion instance
\begin{align*}
    \GG^\bx \coloneqq (H, \bb^\GG, \{f^\bx_u \mid u \in V(G)\}),
\end{align*}
which corresponds to Problem~(\ref{eq:ProxDiff}), i.e. $\EE^{\GG^\bx}(\by) = \Bar{\Phi}_\bx(\by) - \Bar{\Phi}_\bx(\mb{0})$ holds for any feasible potential $\by$.
The exact proximal mapping of $\bx$ is exactly the optimal potential to $\GG^\bx$.

If $\GG$ has total VWF size $S$, so does $\GG^\bx$.
Computing $1+\delta$-approximate optimal potential to $\GG^\bx$ takes $T(|E(H)|, S, 1+\delta)$-time.
\end{claim}

Let $q_\bx$ be an $1+\delta$-approximate optimal potential to $\GG^\bx$.
$q_\bx$ can be seen as an inexact proximal mapping of $\bx$.
The quality of $q_\bx$ yields
\begin{align*}
    \Bar{\Phi}_\bx(q_\bx) - \Bar{\Phi}_\bx(\mb{0}) = \EE^{\GG^\bx}(q_\bx) \le \frac{1}{1+\delta}\EE^{\GG^\bx}(p_\bx) = \frac{1}{1+\delta}\left(\Bar{\Phi}_\bx(p_\bx) - \Bar{\Phi}_\bx(\mb{0})\right) \le 0.
\end{align*}
Let $\eps_\bx$ be $q_\bx$'s additive error, i.e. $\eps_\bx = \Bar{\Phi}_\bx(q_\bx) - \Bar{\Phi}_\bx(p_\bx)$.
Rearranging the above inequality yields:
\begin{align}
\label{eq:addErrBound1}
    0 \le \eps_\bx \le \delta\left(\Bar{\Phi}_\bx(\mb{0}) - \Bar{\Phi}_\bx(q_\bx)\right).
\end{align}

In order to prove \cref{lemma:ProxAGD}, we run \cref{algo:ProxAGD} with $\bx^{(0)} = \mb{0}$ and $T=10\sqrt{\kappa}$ and compute inexact proximal mapping via \cref{claim:ProxDiff} with uniform $\delta$.
The inexactness parameter $\delta$ is determined later.

Next, we analyze the convergence of the algorithm.
Inequality~(\ref{eq:addErrBound1}) and \cref{thm:InexactAPGD} yields
\begin{align*}
    \Phi(\by^{(T)}) - \Phi(\bx^*)
    &\le \frac{2}{(T+1)^2}\left(\norm{\bx^{(0)} - \bx^*} + 3\sum_{k=1}^{T}k \sqrt{2\eps_{k}}\right)^2 \\
    &\le \frac{2}{(T+1)^2}\left(\norm{\bx^{(0)} - \bx^*} + 3\sum_{k=1}^{T}k \sqrt{2\delta\left(\Bar{\Phi}_{\bx^{(k)}}(\bx^{(0)}) - \Bar{\Phi}_{\bx^{(k)}}(\by^{(k)})\right)}\right)^2.
\end{align*}
We need the following claims to further bound the convergence rate.
\begin{claim}
\label{claim:strongCvxPhi}
$\frac{1}{2\kappa}\norm{\bx^{(0)} - \bx^*}^2 \le \Phi(\by^{(0)}) - \Phi(\bx^*)$.
\end{claim}

\begin{claim}
\label{claim:polyDelta}
By setting $\delta = 1/\poly(n)$, we have $\delta\left(\Bar{\Phi}_\bx(\mb{0}) - \Bar{\Phi}_\bx(q_\bx)\right) \le (\Phi(\by^{(0)}) - \Phi(\bx^*))/20000\kappa^4.$
\end{claim}

These 2 claims yields:
\begin{align*}
    \Phi(\by^{(T)}) - \Phi(\bx^*)
    &\le \frac{2}{(T+1)^2}\left(\sqrt{2\kappa \left(\Phi(\by^{(0)}) - \Phi(\bx^*)\right)} + \sum_{k=1}^{T}\frac{k}{100\kappa^2} \sqrt{\left(\Phi(\by^{(0)}) - \Phi(\bx^*)\right)}\right)^2 \\
    &\le \frac{2}{(T+1)^2} \left(\Phi(\by^{(0)}) - \Phi(\bx^*)\right)\left(\sqrt{2\kappa} + \frac{T^2}{100\kappa^2}\right)^2 \\
    &\le \frac{1}{2} \left(\Phi(\by^{(0)}) - \Phi(\bx^*)\right),
\end{align*}
where the last inequality uses that $T=10\sqrt{\kappa}$.

Now we prove \cref{lemma:ProxAGD}:
\begin{proof}[Proof of \cref{lemma:ProxAGD}]
Discussions above reduce diffusion problem to $T=O(\sqrt{\kappa})$ sub-problems, each of them is also a diffusion problem constructed in linear time.

The running time is therefore bounded by:
\begin{align*}
    T(m, S, 2) = O\left(\sqrt{\kappa}\left(m + T\left(|E(H)|, S, 1+\frac{1}{\poly(n)}\right)\right)\right).
\end{align*}
\end{proof}

The rest of the section proves \cref{claim:strongCvxPhi} and \cref{claim:polyDelta}.

\begin{proof}[Proof of \cref{claim:strongCvxPhi}]
Let $p_{\bx^*}$ be the proximal mapping of $\bx^*$.
Using the fact that $\Bar{g}_\bx(\by) \ge g(\by)$ for all $\bx, \by$ yields
\begin{align*}
    \Phi(\bx^*) = \Bar{\Phi}_{\bx^*}(\bx^*) \ge \Bar{\Phi}_{\bx^*}(p_{\bx^*}) \ge \Phi(p_{\bx^*}) \ge \Phi(\bx^*).
\end{align*}
Therefore, $\bx^*$ is one of proximal mapping of itself.

Applying \cref{lemma:inexact_grad_lemma} with $\bu = \by^{(0)}$, $\bx=\bx^*$, and $q_\bx = p_\bx = \bx^*$ yields:
\begin{align*}
    \Phi(\by^{(0)}) \ge \Phi(\bx^*) + \frac{1}{2\kappa}\norm{\by^{(0)} - \bx^*}^2
\end{align*}
since $\bd_\bx = \mb{0}$ and $\eps_\bx = 0$.
\end{proof}




\begin{proof}[Proof of \cref{claim:polyDelta}]

The definition of $\Bar{\Phi}_\bx(\cdot)$ yields
\begin{align*}
    \Bar{\Phi}_\bx(\mb{0}) 
    &= g(\bx) + \grad g(\bx)^\top (\mb{0} - \bx) + \frac{1}{2}\norm{\bx}_{\bL(H)}^2 + h(\mb{0}) \\
    &= \frac{1}{2}\norm{\bx}_{\bL(G)}^2 - \norm{\bx}_{\bL(G)}^2 + \frac{1}{2}\norm{\bx}_{\bL(H)}^2 + h(\mb{0}) \\
    &\le \frac{1}{2}\norm{\bx}_{\bL(H)}^2 + h(\mb{0}) = \frac{1}{2}\norm{\bx}_{\bL(H)}^2 + \Phi(\mb{0}),
\end{align*}
where the last inequality uses that $g(\mb{0}) = 0$.

Combining with the smoothness of $g$ yields
\begin{align*}
    \Bar{\Phi}_\bx(\mb{0}) - \Bar{\Phi}_\bx(q_\bx)
    &\le \Bar{\Phi}_\bx(\mb{0}) - \Phi(q_\bx) \\
    &\le \frac{1}{2}\norm{\bx}_{\bL(H)}^2 + \Phi(\mb{0}) - \Phi(q_\bx).
\end{align*}

Under \cref{assump:polyNum}, we can set $\delta=1/\poly(n)$ such that
\begin{align*}
    \delta\left(\frac{1}{2}\norm{\bx}_{\bL(H)}^2 + \Phi(\mb{0}) - \Phi(q_\bx)\right) \le \frac{\Phi(\mb{0}) - \Phi(q_\bx)}{20000\kappa}.
\end{align*}

We conclude the proof by observing that $\Phi(\mb{0}) - \Phi(q_\bx) \le \Phi(\mb{0}) - \Phi(\bx^*)$.

\end{proof}

\bibliographystyle{alpha}
\bibliography{ref}

\end{document}